\documentclass[english]{article}
\usepackage[T1]{fontenc}
\usepackage[latin9]{inputenc}
\usepackage{geometry}
\usepackage{float}
\usepackage{mathtools}
\usepackage{bm}
\usepackage{amsmath}
\usepackage{amssymb}
\usepackage{graphicx}
\usepackage{babel}
\usepackage{amsthm}
\usepackage{color}
\newtheorem{theorem}{Theorem}

\newtheorem{lemma}[theorem]{Lemma}

\newcommand{\rem}[1]{}
\newcommand{\R}{\mathbb{R}}

\newcommand{\x}{\bm{x}}

\begin{document}
\title{The Harmonic Lagrange Top and the Confluent Heun Equation}
\author{Sean R.~Dawson,  Holger R.~Dullin, Diana M.H.~Nguyen}
	\date{}
\maketitle

\begin{abstract}
The harmonic Lagrange top is the Lagrange top plus a quadratic (harmonic) potential term. 
We describe the top in the space fixed frame using a global description with a Poisson structure on $T^*S^3$.
This global description naturally leads to a rational parametrisation of the set of critical values of the energy-momentum map.
We show that there are 4 different topological types for generic parameter values. 
The quantum mechanics of the harmonic Lagrange top is described by the most general confluent Heun equation (also known as the generalised spheroidal wave equation).  
We derive formulas for an infinite pentadiagonal symmetric matrix representing the Hamiltonian from which the spectrum is computed.
\end{abstract}
\vspace{1ex}
\centerline{Dedicated to the memory of Alexey Borisov.}

\section{Introduction}

The Lagrange top is a prime example of classical mechanics. Over centuries, it has been studied starting with Euler and Lagrange, and interest in its various features is blossoming again and again. Almost every modern development in mechanics has lead to new insights about the Lagrange top.
Before we attempt to describe the place of the Lagrange top in mechanics in the remainder of this introduction, let us formulate our main observation: The quantum mechanics of the harmonic Lagrange top is described by the most general confluent Heun equation (also known as the generalised spheroidal wave equation). By harmonic Lagrange top we mean the Lagrange top with an added harmonic (i.e. quadratic) potential. It provides an example of the subcritical and the supercritical Hopf bifurcation and its quantisation. The bulk of the paper is devoted to the description of the classical integrable system.

Rigid body dynamics is treated in most mechanics textbooks, e.g.~\cite{Whittaker37,LanLif84,Arnold78,Gold80,MarsdenRat94}.
Of the books devoted specifically to rigid body dynamics we highlight the monumental volumes of Klein \& Sommerfeld \cite{KS10} 
and the recent addition by 
Borisov \& Mamayev \cite{Borisov18}.
Many special cases of rigid body dynamics including the Lagrange top are completely integrable Hamiltonian systems, and as such have been studied in detail in
Bolsinov \& Fomenko \cite{BolFom04} and Cushman \& Bates \cite{CushmanBates15}. For all the references we inevitably missed in this introduction we refer to the extensive bibliography in \cite{Borisov18}.

In modern mechanics the (energy)-momentum map plays a central role.  Singularity theories' swallowtail was found as the set of critical values of the energy-momentum map of the Lagrange top in 
\cite{Cushman85}, also see \cite{CushmanBates15}. 
The meaning of the swallowtail from the point of view of bifurcation theory, specifically the subcritical Hopf bifurcation in the Lagrange top was described in
\cite{cushman90}. 
The fact that the swallowtail may make the image of the energy-momentum map non-simply connected is the essential observation that explains why it does not possess global action variables
\cite{Duistermaat80,CushDuist88}.
Hamiltonian monodromy of the Lagrange top is described in \cite{vivolo2003monodromy}.
Integrable discretisations of the integrable Lagrange top were found in \cite{bobenko99}.
The complex algebraic geometry of the Lagrange top was described in \cite{Gavrilov98}, and its bi-Hamiltonian structure in \cite{Tsiganov08}.
In KAM theory perturbations of the Lagrange top give a beautiful example worked out in detail in \cite{Hanssmann06}.

The quantisation of the symmetric top was first done in the early days of quantum mechanics
\cite{Reiche26quantelung}
and leads to a hypergeometric equation,
also see 
\cite{landau13quantum}.
The study of polar molecules in an electric field leads to a Hamiltonian that is equivalent to the Lagrange top. In the physics literature this is referred to as the Stark effect, and was first studied in
\cite{Schlier55}.
Matrix elements for the numerical computation of the spectrum were given in 
\cite{Shirley63stark}, and nearly 30 years later again in
\cite{Hajnal91stark}.

The discovery of quantum monodromy \cite{Duistermaat80} was in the smaller brother of the Lagrange top, the spherical pendulum, in 
\cite{CushDuist88}.
The quantum monodromy in the Lagrange top itself has been studied in \cite{Kozin03monodromy}.

While so-called semi-toric systems with two degrees of freedom (somewhat like the spherical pendulum) are now in a precise sense completely understood classically \cite{VuNgoc09} and quantum mechanically, \cite{SanVuNgoc21} the Lagrange top is still out of reach from this point of view.
We should mention that many generalisations of the spherical pendulum have been studied, in particular the magnetic spherical pendulum \cite{cushman1995magnetic,CushmanBates15}, also see \cite{saksida2002neumann}), and the quadratic spherical pendulum \cite{zou1992kolmogorov,Efstathiou05}. The combination of both is the harmonic Lagrange top, which is the object of this paper. To our knowledge, it has not been considered in the literature. It is an example of the general idea described in \cite{DP15}, where semi-toric systems are deformed preserving integrability.
In particular, we find that the harmonic Lagrange top exhibits the subcritical and the supercritical Hopf bifurcations for certain parameters.


As mentioned in the beginning, we want to draw attention to the fact that the quantisation of the Lagrange top leads to the confluent Heun equation. 
The Heun equation is a Fuchsian equation with 4 regular singular points, thus generalising the hypergeometric equation by one singularity, see, e.g., \cite{arscott64,ronveaux95,slayvanov00,DLMF}. 
An important physical application of the confluent Heun equation 
appears in the perturbation theory of a rotating black hole in 
general relativity \cite{teukolsky1973perturbations,press1973perturbations,Leaver86}. In this context, expansions in terms of Jacobi polynomials have been given in \cite{Crossman77}, and series expansion for small potential are given in \cite{seidel1989comment}.
As we show below, the harmonic Lagrange top leads to the most general confluent Heun equation, unlike the above application in general relativity, which does not have enough parameters.

The structure of this paper is as follows. We give an introduction to the Lagrange top in the next section, where we emphasise the description in the spatial frame using quaternions and the corresponding Poisson structure. The various periodic flows and their differences when considering $T^*SO(3)$ or $T^*S^3$ (the quaternions) is discussed in section 3, and the reductions to two degrees of freedom in section 4. The traditional description in Euler angles is recalled in section 5, which is needed for the quantisation. The main classical results are the description of the critical points in phase space and the corresponding critical values in the image of the energy-momentum map. There are 4 different cases, with one thread (the original Lagrange top), with two threads, with a triangular tube instead of the thread, and a triangular tube shrinking to a thread. In the final section we show that the quantum harmonic Lagrange top leads to the most general confluent Heun equation and compute the spectrum, which is displayed overlayed with (slices) of the classical energy-momentum map. A new method for the computation of the spectrum is presented.

\section{Heavy Symmetric Top}

Consider a general rigid body with a fixed point. Assume that the symmetric inertia tensor $I$ with respect to that point 
has three distinct eigenvalues $I_1$, $I_2$, $I_3$, the moments of inertia, and assume that a body frame has been chosen in which the tensor of inertia is diagonal. 
For the symmetric top with $I_1 = I_2$ the location of the corresponding basis vectors is only defined up to a rotation about the symmetry axis (or figure axis) of the body. 
In in the spatial coordinate frame, the $z$-axis is parallel to the direction of gravity. 
Let $\bm{V}$ be the coordinate vector of a point in the body frame. The orthogonal matrix $R \in SO(3)$ describes how this point is moving in time when viewed in the spatial frame, $\bm{v} = R \bm{V}$.

For the free rigid body (Euler top), the fixed point of the body is the centre of gravity of the body. For the Lagrange top, the centre of gravity is on the figure axis.
Denote the unit vector along the figure axis of the top by $\bm{a}$ (in the spatial frame), then the potential energy in the field of gravity is $V = c_1 a_z$. In this paper, we are going to study the more general case
\[
      V(a_z) =c_1 a_z + c_2 a_z^2
\]

The angular velocity $\bm{\Omega}$ in the body frame is defined through $R$ by 
$R^t \dot R \bm{V} = \Omega \times \bm{V}$ for any vector $\bm{V}$, or, equivalently, by $\hat{\bm{\Omega}} = R^t \dot R$.
The kinetic energy of the rigid body is 
\[
    T = \frac12 \bm{\Omega} \cdot I \bm{\Omega}
\]
where $I$ is the diagonal tensor of inertia.

The angular momentum vector is defined by $\bm{L} = I \bm{\Omega}$. For the free rigid body $\bm{l} = R \bm{L}$ is a constant vector. For the Lagrange top instead there are only two conserved quantities given by
\[
   l_z = \bm{l} \cdot \bm{e}_z, \quad
   L_3 = \bm{L} \cdot \bm{e}_3 = R^t \bm{l} \cdot \bm{e}_3 = \bm{l} \cdot R \bm{e}_3 = \bm{l} \cdot \bm{a} \,.
\]
We use both $\bm{e}_z$ and $\bm{e}_3$ as abbreviations for $(0,0,1)^t$.

A beautiful global description of the dynamics of rigid bodies uses quaternions $\bm{x} = (x_0, x_1, x_2, x_3)$ which are coordinates on the double cover of $SO(3)$ which is $S^3 \in \R^4$ given by $x_0^2 + x_1^2 + x_2^2 + x_3^2 = 1$. Define 
\[
    \bm{x}_\pm = \begin{pmatrix}
    x_1 & -x_0 & \mp x_3 & \pm x_2 \\
    x_2 & \pm x_3 & -x_0 & \mp x_1 \\
    x_3 & \mp x_2 & \pm x_1 & -x_0
    \end{pmatrix}
\]
which satisfy $\x_+ \x_+^t = id$, $\x_- \x_-^t = id$, $\x_+^t \x_+ \x_-^t = \x_-^t$, and $\x_-^t \x_- \x_+^t = \x_+^t$ on the unit sphere.
Then an orthogonal $3\times 3$ matrix is given by $R = \bm{x}_+ \bm{x}_-^t$. 
The last two identities in the previous sentence become $\x_+^t R = \x_-^t$ and $\x_-^t R^t = \x_+^t$.
The matrices $\bm{x}_\pm$ relate
the angular velocities to the tangent vector of the sphere $\dot{\bm{x}}$ by 
$\bm{\Omega} = 2 \bm{x}_- \dot{\bm{x}}$ and 
$\bm{\omega} = 2\bm{x}_+ \dot{\bm{x}}$, see, e.g., \cite[Section~16]{Whittaker37}.
To see this, differentiate $R$ with respect to time, observe that $\dot{\bm{x}}_+ \bm{x}_-^t = \bm{x}_+ \dot{\bm{x}}_-^t$, and use $R^t \bm{x}_+ = \bm{x}_-$.
Substituting $\bm{\Omega} = 2\bm{x}_- \dot{\bm{x}}$ into the expression for $T$ gives
\[
T=2 \dot{\bm{x}}^t (\bm{x}^t_- I \bm{x}_-) \dot{\bm{x}} \,.
\]
Differentiating with respect to $\dot{\bm{x}}$ gives the conjugate momenta $\bm{p}=4(\bm{x}^t_- I \bm{x}_-) \dot{\bm{x}}$ on $T^*S^3$. Using $\bm{L}= I\bm{\Omega}=2 I\bm{x}_- \dot{\bm{x}}$ we see that 
\[
\bm{L}=2 I\bm{x}_- \dot{\bm{x}}=2(\bm{x}_-\bm{x}^t_-) I \bm{x}_- \dot{\bm{x}}=\frac{1}{2}\bm{x}_- \bm{p} \,.
\]
Similarly, we have $\bm{l}=\frac12 \bm{x}_+\bm{p}$.
It is valid to use the canonical bracket between $\bm{x}$ and $\bm{p}$ because the resulting Hamiltonian automatically preserves $|\bm{x}|=1$ and $\bm{x} \cdot \bm{p} = 0$.

Now changing from canonical variables $(\bm{x}, \bm{p})$ to non-canonical variables $(\bm{x}, \bm{L})$ gives the Lie-Poisson structure in the body frame as \cite{Borisov97,Borisov18}
\[
     B_- = \begin{pmatrix}
     0 & \tfrac12 \bm{x}_-^t \\
     -\tfrac12 \bm{x}_- & \hat{\bm{L}}
     \end{pmatrix}, \quad 
     \dot{\bm{x}} = \tfrac12 x_-^t \nabla_L H, \quad 
     \dot{\bm{L}} = -\tfrac12 x_- \nabla_x H + \bm{L} \times \nabla_L H\,.
\]
Similarly, the Lie-Poisson structure in the space fixed frame is 
\[
     B_+ = \begin{pmatrix}
     0 & \tfrac12 \bm{x}_+^t \\
     -\tfrac12 \bm{x}_+ & -\hat{\bm{l}}
     \end{pmatrix}, \quad
     \dot{\bm{x}} = \tfrac12 x_+^t \nabla_l H, \quad \dot{\bm{l}} = -\tfrac12 x_+ \nabla_x H - \bm{l} \times \nabla_l H \,.
\]
Both Poisson structures have the Casimir $x_0^2 + x_1^2 + x_2^2 + x_3^2$. 
The Poisson structure $B_+$ is found by sandwiching the symplectic structure of the $(\x, \bm{p})$ variables by the Jacobian of the transformation of $(\x, \bm{l})$ and its transpose.

For the Euler top the usual Hamiltonian in the body frame is $H = \frac12 \bm{L} \cdot  I^{-1} \bm{L}$, and the complicated integrals are $R\bm{L}$ (which imply the simple integral $|\bm{L}|^2$). In the space fixed frame instead we have the complicated Hamiltonian $H = \frac12 \bm{l} \cdot R I^{-1} R^t \bm{l}$ with the simple integrals $\bm{l}$. We mention the Euler top here to make the point that for general moments of inertia, the description in the body frame is simpler. However, for a round rigid body with $I_1=I_2=I_3$ both Hamiltonians are equally simple. Also for a symmetric rigid body with say $I_1 = I_2$, the spatial frame is useful because 
\[
 2T= \bm{l} \cdot RI^{-1}R^t \bm{l} = \bm{l} \cdot \frac{1}{I_1} R( id + \delta \bm{e}_3 \bm{e}_3^t)R^t \bm{l} = \frac{1}{I_1}\left( \bm{l}^2 + \delta L_3^2 \right)
\]
where $\delta=I_1/I_3-1$. The important point is that $L_3 = \bm{e}_3 \cdot \bm{L} = \bm{e}_3 \cdot R^t \bm{l}=R\bm{e}_3 \cdot\bm{l} = \bm{l} \cdot \bm{a}$ is the angular momentum about the body's symmetry axis $\bm{e}_3$ and hence a constant of motion for the symmetric top.

\begin{theorem}
The Lagrange top (symmetric heavy rigid body with a fixed point on the symmetry axis) in coordinates $\bm x \in S^3 \subset \R^4$ and angular momenta $\bm{l}$ in the space fixed frame has Hamiltonian 
\[
H = \frac{1}{2I_1}( l_x^2 + l_y^2 + l_z^2 + \delta L_3^2) + V(x_0^2 + x_3^2 - x_1^2 - x_2^2)
\]
and Poisson structure $B_+$, with integrals $l_z$ and
\[
L_3 = 2l_x(-x_0x_2+x_1x_3) + 2l_y(x_0x_1+x_2x_3) + l_z (x_0^2+x_3^2-x_1^2-x_2^2).
\]
The vector fields of $l_z$ and $L_3$ 
generate a $T^2$ action with singularities. The vector field of the Hamiltonian is 
\begin{equation} \label{eqn:XH}
X_H = 
   \frac{1}{2 I_1} X_{l^2} 
+ \frac{ \delta L_3}{I_1}  X_{L_3} 
- \frac12 (0,0,0,0, \bm{x}_+ \nabla_x V)^t \,.
\end{equation}
The functions $H$, $l_z$, $L_3$ have pairwise vanishing Poisson bracket.
The vector fields $X_H$, $X_{L_3}$ and $X_{l_z}$ are independent almost everywhere.
\end{theorem}

This theorem is well known for the case of linear potential, and when using Euler angles it is part of most mechanics text books. Instead we offer a global description in the spatial frame with a Poisson structure. In addition, in order to make the connection with the general confluent Heun equation, we consider not just a linear potential (gravity), but in addition a quadratic term. 
After some preparations in the next sections discussing the torus action, the reduction, and briefly recalling Euler angles, the main technical part is the description of the set of critical values of the energy-momentum map in Lemma~2.

\section{Torus action}

The vector field generated by $L_3$ in the space fixed coordinate system is
\begin{equation} \label{eqn:XL3}
X_{L_3} = B_+ \nabla L_3 
= \tfrac12( \bm{x}_+^t R\bm{e}_3,0,0,0 )^t
 =(\tfrac12 \bm{x}_-^t \bm{e}_3,0,0,0)^t
\end{equation}
where we used the identity $\x_+^t R = \x_-^t$.
This vector field can be easily integrated (two harmonic oscillators) to give the flow $\Phi_{L_3}^\psi$. This flow rotates $(x_0, x_3)$ and $(x_1,x_2)$ by $\psi/2$ clockwise.
However, when the flow acts on $R$ it acts by multiplication by a counterclockwise rotation about the $z$-axis through $\psi$ (not $\psi/2$!) from the right.
Thus $L_3$ has $2\pi$-periodic flow on $T^*SO(3)$ and hence is an action variable.

The vector field generated by the integral $l_z$ is 
\begin{equation} \label{eqn:Xlz}
X_{l_z} = B_+ \nabla l_z  = ( \tfrac12 \bm{x}_+^t \bm{e}_z,-\bm{l} \times \bm{e}_z )^t.
\end{equation}
Again, this vector field is easily integrated (three harmonic oscillators) giving the flow $\Phi_{l_z}^\phi$.
The action on $R$ is by multiplication with a counterclockwise rotation about the $z$-axis through $\phi$ from the left. In addition, the momentum vector $\bm{l}$ is rotated by the same rotation matrix.
Thus $l_z$ has $2\pi$-periodic flow on $T^*SO(3)$ and hence is an action variable.

The vector fields $X_{L_3}$ and $X_{l_z}$ are parallel when $l_x = l_y = 0$ and either $x_0 = x_3 = 0$ or $x_1 = x_2 = 0$. These critical points have $\bm{l} || \bm{a}|| \bm{e}_z$ and are called sleeping tops. In the first case $a_z = -1$ (hanging sleeping top), while in the second case $a_z = +1$ (upright sleeping top). The torus action is singular at these points because the rotations coincide. Since $L_3 = \bm{l} \cdot \bm{a}$ we see that $L_3 = -l_z$ for the hanging sleeping top and $L_3 = l_z$ for the upright sleeping top. The corresponding critical points of $H$ are two parabolas above $l_z \pm L_3 = 0$.

\rem{\footnote{When considering the flows of the same integrals $l_z$ and $L_3$ but in the body frame with Poisson structure $B_-$, the action on the rotation matrix $R$ is the same, while the action on the momenta is reversed: The flow generated by $L_3$ (in the body frame) rotates the momentum vector $\bm{L}$ clockwise, while the flow generated by $l_z$ (in the body frame) leaves the momentum vector $\bm{L}$ fixed.}
}

The vector fields $X_{l_z}$ and $X_{L_3}$ both have $2\pi$ periodic flows on $T^*SO(3)$, i.e. they map $\bm{x}$ to $-\bm{x}$ after time $2\pi$.
 When considered as flows on $S^3$ both flows have period $4\pi$.
%
Now consider the vector fields generated by $l_z \pm L_3$. These are both $2\pi$ periodic vector fields on $T^*S^3$. 
Points with $l_x = l_y = 0$ and either $x_0 = x_3 = 0$ or $x_1 = x_2 = 0$, respectively, are fixed points of these flows. Nevertheless, they are  action variables on $T^*S^3$.
Notice that as flows on $T^*SO(3)$ the flows of $l_z \pm L_3$ do not have constant period, since points with $l_x=l_y=0$ and either $x_0 = x_3$ or $x_1=x_2$ have minimal period $\pi$, while all other non-fixed points have minimal period $2\pi$.
The $T^2$ action on $T^*S^3$ is of course still singular, the  difference is that now the exceptional sets of points are found as those where one of the vector fields vanishes.

The vector field of the spherical Euler top is that of $\bm{l}^2 = l_x^2 + l_y^2 + l_z^2$. The vector fields of $l_x$ and $l_y$ are permutations to that of $l_z$ given in \eqref{eqn:Xlz}. Combining these gives
\[
  X_{l^2} = (\bm{x}_+^t \bm{l} , 0,0,0)^t \,.
\]
Here the components of $\bm{l}$ are all constant, and the flow of this vector field is a rotation about the axis $\bm{l}$. This is also a periodic flow, but the period is not constant. To obtain constant period, we consider the flow generated by $l = \sqrt{ \bm{l}^2}$, which we denote by $X_l$.
This flow commutes with the flows of $l_z$ and $L_3$, but not with that of $H$.
The flow of $\bm{l}^2$ leaves $\bm{l}$ constant and so
\[
   \Phi_{l}^\alpha = \exp \left( \frac{\alpha}{2 l}
\left(
\begin{array}{cccc}
 0 & l_x & l_y & l_z \\
 -l_x & 0 & -l_z & l_y \\
 -l_y & l_z & 0 & -l_x \\
 -l_z & -l_y & l_x & 0 \\
\end{array}
\right) \right)
=
\left(
\begin{array}{cccc}
 \cos  \frac{\alpha }{2}  & l_x/l \sin  \frac{\alpha }{2}  & l_y/l \sin \frac{\alpha }{2}  & l_z/l \sin  \frac{\alpha}{2}  \\
 -l_x/l \sin  \frac{\alpha }{2}  & \cos  \frac{\alpha }{2}  & -l_z/l \sin \frac{\alpha }{2}  & l_y/l \sin  \frac{\alpha}{2}  \\
 -l_y/l \sin \frac{\alpha }{2}  & l_z/l \sin  \frac{\alpha }{2}  & \cos  \frac{\alpha }{2}  & -l_x/l \sin  \frac{\alpha}{2}  \\
 -l_z/l \sin  \frac{\alpha }{2}  & -l_y/l \sin  \frac{\alpha }{2}  & l_x/l \sin  \frac{\alpha }{2}  & \cos  \frac{\alpha}{2}  \\
\end{array}
\right) \,.
\]
When acting with this flow on the rotation matrix $R$ with initial condition $\bm{x} = (1,0,0,0)$ gives Rodrigues' parametrisation of $SO(3)$ with rotation axis $\bm{l}/l$ and rotation angle $\alpha$. Thus Rodrigues' formula gives the geodesics of the spherical top.
When acting with this flow on $S^3$ it is periodic with period $4\pi$.

The reason we are including this flow is that there is an interesting difference between $SO(3)$ and $S^3$. On $T^*SO(3)$ the singular $T^3$ torus action generated by the commuting flows of $l_z$, $L_3$, and $l$ is faithful. This means that outside the singularity where $l_z \pm L_3 = 0$ the action on each $T^3$ obtained by fixing the values of the generators is faithful and effective.
By contrast, when considering the $T^3$ torus action generated by $l_z+L_3$, $l_z - L_3$, and $l+l$ on $T^*S^3$ the action is not faithful on regular tori. The reason is that when flowing each flow only for angle $\pi$, then the first two flows together achieve $\x \to -\x$, and this is cancelled by the flow $\Phi_{2l}^\pi=\Phi_l^{2\pi}$. 

\section{Reductions}

The flows of $l_z$ and $L_3$ are global $S^1$ actions, and hence allow for regular reduction. It is straightforward to obtain the reduced system from the global system with Poisson structure $B_\pm$. The $l_z$-reduced system are the well known Euler-Poisson equations, while the $L_3$-reduced equations are somewhat less well known in classical mechanics (see, e.g., \cite{bobenko99,CushmanBates15,D03b}).
The full reduction is singular because the $T^2$ action of $l_z$ and $L_3$ is singular. The standard description is in $zxz$-Euler angles, the singular reduction using invariants is in \cite{CushmanBates15}. A peculiar property of Euler angles 
is that the $\psi$-rotation leaves the figure axis invariant (it acts on the right) while the $\phi$-rotation leaves the direction of gravity invariant (it acts on the left), and hence Euler angles are neither space-fixed nor body-fixed.
The quantisation of the top (see below) starts out with Euler angles \cite{landau13quantum}, but in the end, writing the Hamiltonian using $\bm{l}^2$ and $L_3^2$ suggests that there 
 the spatial frame is also useful.

The reduction by the symmetry $\Phi_{L_3}^\psi$ introduces the coordinates of the axis of the top $\bm{a} = R \bm{e}_3$ as new coordinates. 
This is, in fact, reduction by invariants, since 
the third column of $R$ is given by 
$( 2(x_0 x_2 + x_1 x_3), -2 x_0 x_1 + 2 x_2 x_3, x_0^2 + x_3^2 - x_1^2 - x_2^2 )$ and these are all 
invariant under the two-oscillator flow $\Phi_{L_3}^\psi$. 
We already noted that $\Phi_{L_3}^\psi$ acts on $R$ by multiplication by $R_z(\psi)$ from the right, and hence
$R R_z(\psi) \bm{e}_3 = R \bm{e}_3 = \bm{a}$ is invariant.
The resulting reduced system has Poisson structure
\[
B_+^r = \begin{pmatrix}
   0 & -\hat{\bm{a}}\\
   -\hat{\bm{a}} & -\hat{\bm{l}}
\end{pmatrix}, 
\quad
\dot{\bm{a}} = - \bm{a} \times \nabla_l H, \quad
\dot{\bm{l}} = - \bm{a} \times \nabla_a H - \bm{l} \times \nabla_l H
\,.
\]
Denote the Jacobian of the transformation from $(\x, \bm{l})$ to $(\bm{a}, \bm{l})$ by $A$. 
Then $B_+^r = A^t B_+ A$ when expressed in the new variables.
The main identity in the reduction from $B_+$ to $B_+^r$ is
$\frac12 \frac{ \partial \bm{a}}{\partial \bm{x}} \bm{x}_+^t = \hat{\bm{a}}$.
The Poisson structure $B_+^r$ has
Casimirs $\bm{a}^2=1$ and $\bm{a}\cdot \bm{l}$ and the reduced Hamiltonian is
\[
   H = \frac{1}{2I_1} (  \bm{l}^2 + \delta (\bm{a} \cdot \bm{l})^2) + V(a_z).
\]
Since $\bm{a}\cdot\bm{l}$ is a Casimir (equal in value to the generator of the symmetry $L_3$) it does not contribute to the dynamics but merely changes the value of the Hamiltonian.

\rem{
Choosing the units of angular momentum to be $I_1$ and ignoring the shift in value of the Hamiltonian by a Casimir we arrive at
\[
   H = \frac12 \bm{l}^2 + c a_z + d a_z^2
\]
where we added an extra quadratic potential term. This is the Hamiltonian of the Lagrange top plus harmonic potential that we are going to study in this paper. }

Note that reduction by the symmetry generated by the integral $l_z$ is more complicated in the spatial frame since the flow is a rotation in $\bm{x}$ and in $l_x, l_y$. However, when switching to the body frame then the flow of $l_z$ (written in terms of $\bm{L}$) is simpler. Reduction is achieved by introducing the invariant of the left action generated by $l_z$, which is $ \bm{e}_3^t R_z(\phi) R = \bm{e}_3^t R = \bm{\Gamma}^t$ with Poisson structure
\[
    B_-^r = 
    \begin{pmatrix}
   0 & \hat{\bm{\Gamma}}\\
   \hat{\bm{\Gamma}} & \hat{\bm{L}}
\end{pmatrix},
\quad
\dot{\bm{\Gamma}} = \bm{\Gamma} \times \nabla_L H, \quad
\dot{\bm{L}} = \bm{\Gamma} \times \nabla_\Gamma H + \bm{L} \times \nabla_L H
\,.
\]
The reduction leads to the more familiar Hamiltonian of the Lagrange top given by
\[
    H = \frac{1}{2I_1} ( L_1^2 + L_2^2) + \frac{1}{2I_3} L_3^2 + V( \Gamma_3)
\]
where $\bm{\Gamma}$ is $\bm{e}_z$ viewed from the body frame. The Poisson structure is $B_-^r$ with the opposite sign than $B_+^r$.
These are the equations usually called Euler-Poisson equations. Their advantage is that this reduction remains valid for an arbitrary rigid body with a fixed point, and this family for appropriate moments of inertia and position of the centre of mass contains the Kovalevskaya top, the Euler top, and all other non-integrable tops.

The Hopf bifurcation in the sleeping top with $\bm{a} || \bm{l} || \bm{e}_z$ respectively $\bm{\Gamma} || \bm{L} || \bm{e}_3$ is best described in the reduced system, because the corresponding periodic orbit becomes a relative equilibrium after reduction.
It is easy to check that in deed these are equilibria, and linearising the Hamiltonian vector field about these equilibria yields a $6\times 6$ matrix with 2 eigenvalues zero corresponding to the two Casimirs. The characteristic polynomial for the remaining non-trivial eigenvalues is
\[
    P_+(\lambda) = \lambda^4 + \lambda^2( \kappa^2 - 2 f) + f^2 = 0, \quad 
    \kappa = l_z/I_1 = \omega I_3/I_1, \quad 
    f = a_z V'(a_z)/I_1, \quad 
    a_z = \pm 1 \,.
\]
in the spatial frame and 
\[
   P_-(\lambda) = P_+(\lambda) + \omega( \omega -\kappa) ( 2 \lambda^2 + 2 f + \omega(\omega - \kappa ) ), \quad
   \omega = l_z/I_3
\]
in the body frame. The eigenvalues are not the same because the system is described in a frame rotating with angular velocity $\omega$. However, they differ only by $\pm i \omega$ such that the Floquet multiplier $\mu = \exp( \lambda T)$ of the periodic orbit with period $T = 2\pi/\omega$ is the same.
The description in the spatial frame gives simpler formulas.

At the Hopf bifurcation the eigenvalues change from all purely imaginary via a collision on the imaginary axis to a quadruple of complex eigenvalues. This occurs when the discriminant of $P_+(\lambda)$ considered as a quadratic equation in $\lambda^2$ changes from positive to negative. The discriminant is given by $\kappa^2 (\kappa^2 - 4 f)$. When $f$ is negative the eigenvalues are purely imaginary for any $\kappa$. When $f$ is positive eigenvalues are purely imaginary when $\kappa^2 > 4 f$, while the top is unstable when $\kappa^2 < 4 f$. This is the classical stability condition for the Lagrange top, here obtained for arbitrary potential. At the critical case $\kappa^2 = 4f$ the eigenvalues collide and $\lambda^2 = -\kappa^2/4$. 

\section{Euler Angles}

The Poisson structures $B_\pm$ allow for a global description of rigid body dynamics free of coordinate singularities. However, often explicit canonical coordinates are more convenient, and even essential for the quantisation of the problem.
Such a coordinate system adapted to the symmetries is given by $zxz$-Euler angles such that 
\[
    R = R_z(\phi) R_x(\theta) R_z(\psi) \,.
\]
The canonically conjugate momenta are denoted by $p_\phi$, $p_\theta$, $p_\psi$, respectively. Then we have that $l_z = p_\phi$ and $L_3 = p_\psi$.
The Hamiltonian in these coordinates is
\[
   H = \frac{1}{2I_1} ( 2 T_{round} + \delta p_\psi^2) + V(\cos\theta)
\]
where $T_{round}$ is the kinetic energy of the spherical top with moment of inertia 1:
\[
    T_{round} = \frac12 \left( p_\theta^2 + \frac{1}{\sin^2\theta} \left( p_\phi^2 + p_\psi^2 - 2 p_\phi p_\psi \cos\theta \right)
    \right) = \frac12 \bm{l}^2 \,.
\]
Notice that this round metric on $SO(3)$ is a metric of constant curvature and hence up to a covering equivalent to the metric of the round sphere $S^3$.
Transforming $\phi = (\phi_+ - \phi_-)/2$, and $\psi = (\phi_+ + \phi_-)/2$ with Jacobian $1/2$
in the new coordinate system the metric is diagonal 
\[
    ds^2 = d\theta^2 + \cos^2\tfrac{\theta}{2} \, d\phi_+^2 + \sin^2 \tfrac{\theta}{2} \, d\phi_-^2\,.
\]
This is the metric of the round sphere in doubly cylindrical (or polyspherical, see \cite[10.5]{vilenkin1992representationII}) coordinates
\[
    ((\cos\phi_+, \sin\phi_+)\cos\tfrac{\theta}{2},
    (\cos\phi_-, \sin\phi_-)\sin\tfrac{\theta}{2}),
\]
where $ 0 \le \theta/2 \le \pi/2$ while $\phi_+, \phi_-$ are true angles.

Away from the coordinate singularity where the torus action becomes singular, Euler angles are a smooth local coordinate system. Equilibrium points in $\theta$ are determined by $\partial H/ \partial \theta = 0$. For later use, we denote this function by $H_\theta$, and similarly 
the 2nd derivative by $H_{\theta\theta}$.


\begin{figure}
\begin{centering}
\includegraphics[width=6cm,height=6cm,keepaspectratio]{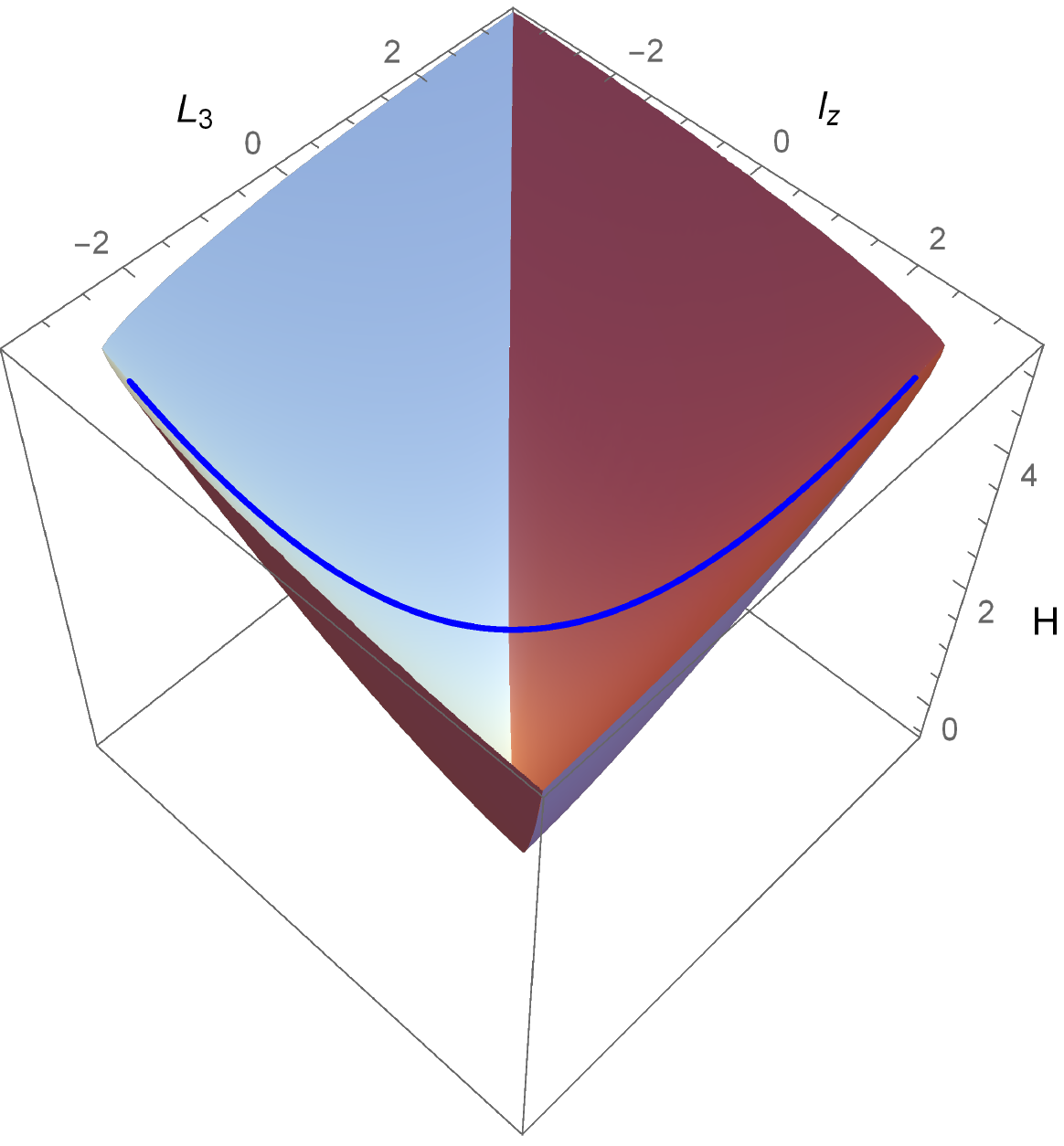}\hspace{1cm}\includegraphics[width=6cm,height=6cm,keepaspectratio]{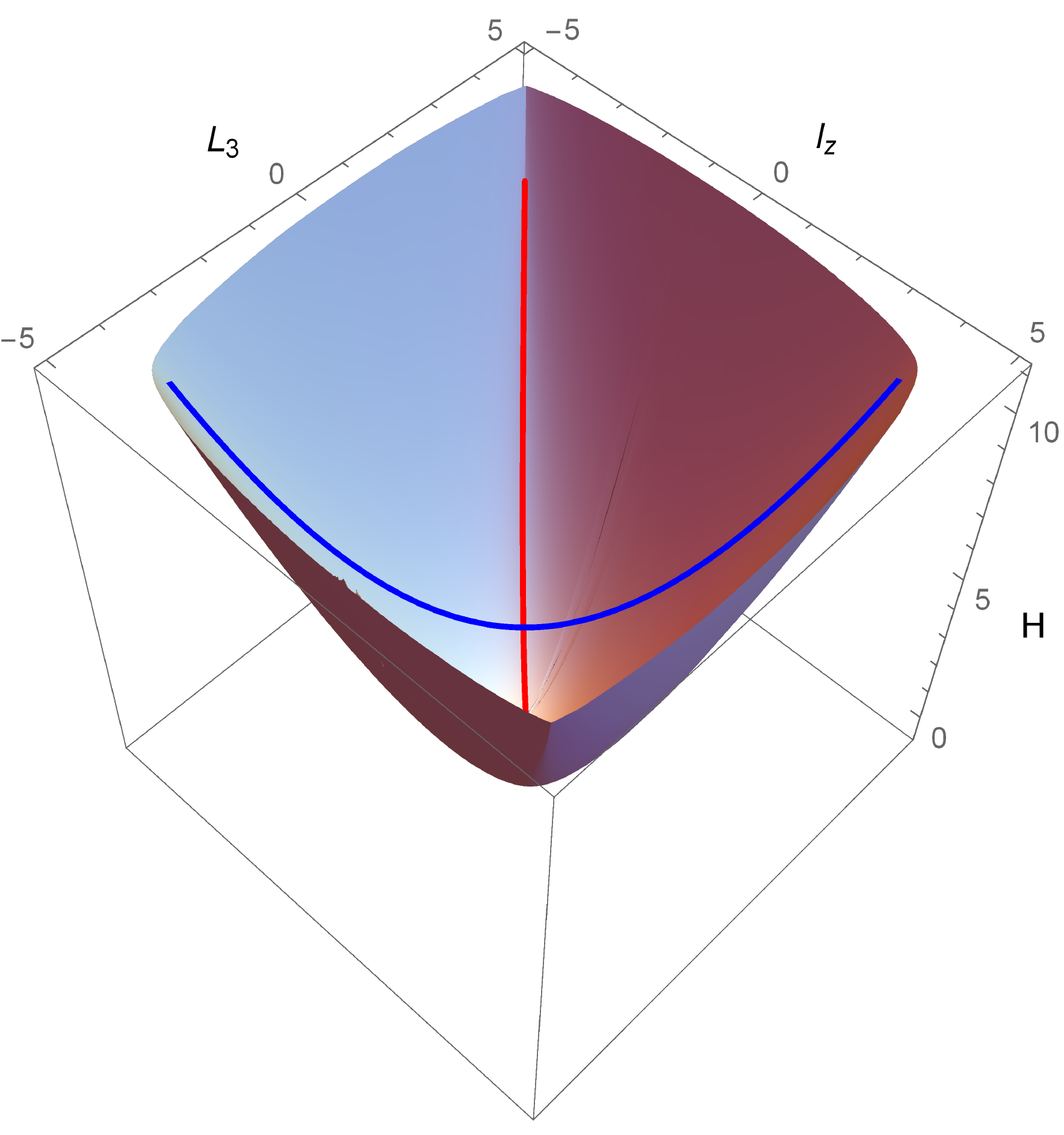}\vspace{0.5cm}\\
\includegraphics[width=6cm,height=6cm,keepaspectratio]{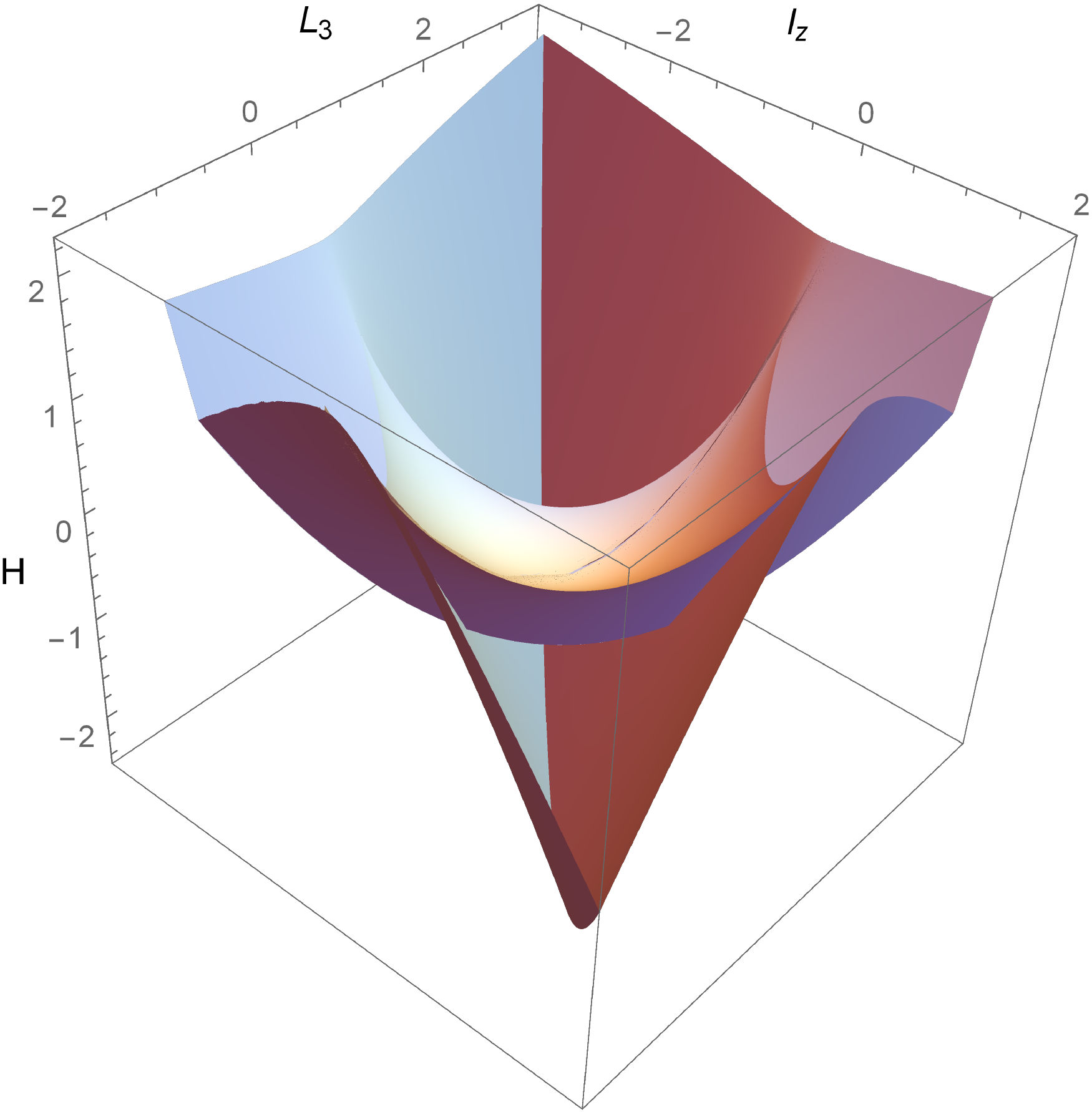}\hspace{1cm}\includegraphics[width=6cm,height=6cm,keepaspectratio]{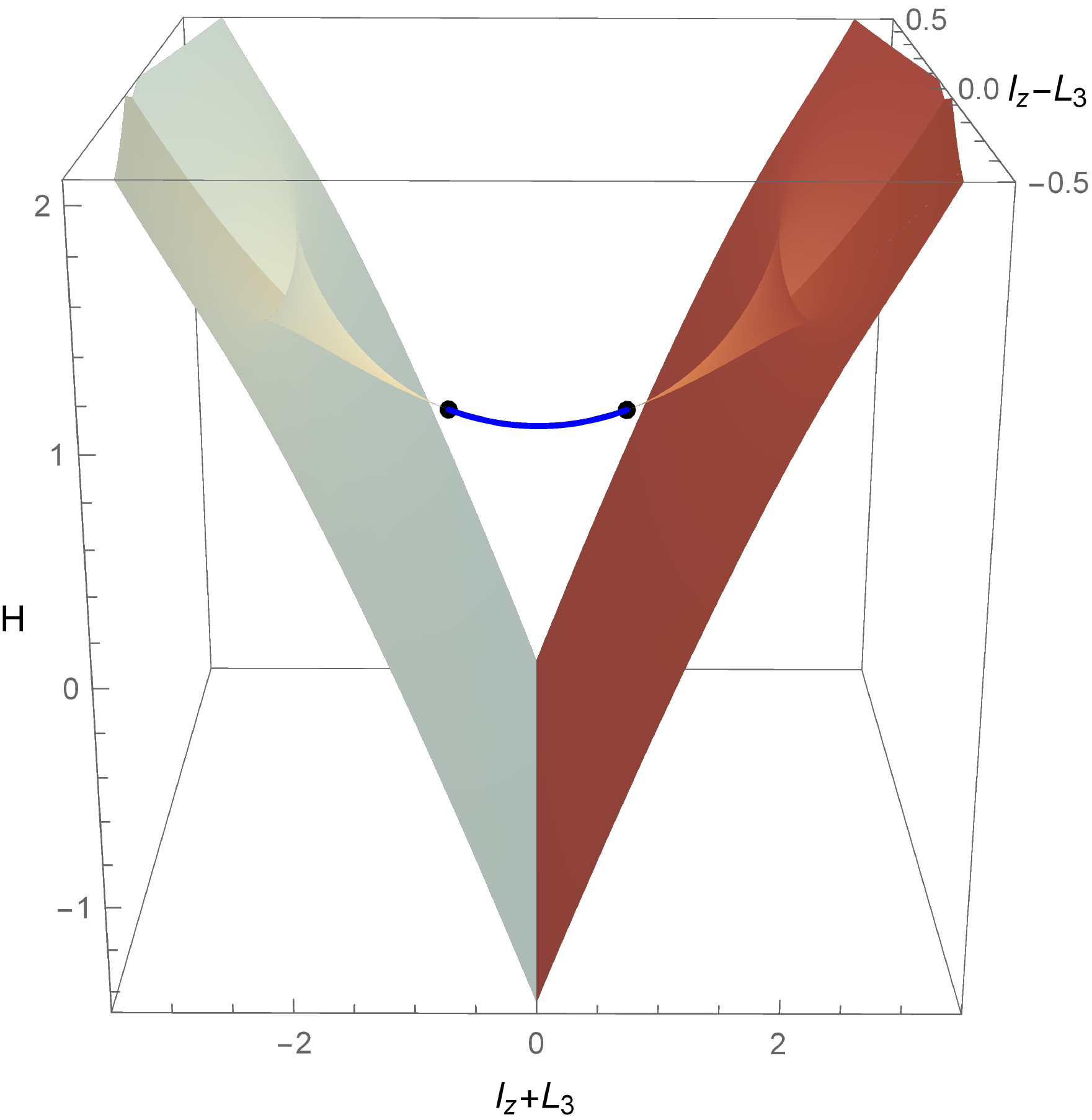}
\par\end{centering}
\caption{Example of a) one thread, b) two threads, c) triangular tube and d) shrinking triangular tube.
All examples use $I_{1}=1,\delta=0,c_{1}=1$. The values of $c_2$ used in a)-d) are $c_{2}=(0.4,2.5,-1.5,-0.48)$. \label{fig:3d examples}}
\end{figure}

\rem{
\begin{figure}[H]
\begin{centering}
\includegraphics[width=6cm,height=6cm,keepaspectratio]{"ThickTube".pdf}\includegraphics[width=6cm,height=6cm,keepaspectratio]{"Pinched Tube Zoom".pdf}\\
\par\end{centering}
\caption{a) Zoomed image of the tube in Figure (\ref{fig:3d examples}) c).
b) In the limit $c_{1}=-2c_{2}$ (with $c_{2}\le0$ and $c_{1}>0$)
the tube has a pinch develop at $(M,K,H)=(0,0,c_{1}+c_{2})$. }
\end{figure}
}

\rem{
\begin{figure}[H]
\begin{centering}
\includegraphics[width=6cm,height=6cm,keepaspectratio]{RedSpecExample\string".pdf}\includegraphics[width=6cm,height=6cm,keepaspectratio]{BlueSpecExample\string".pdf}\\
\par\end{centering}
\caption{a) Slice of $M=0$ and corresponding quantum spectrum. In the limit
$c_{1}=-2c_{2}$ the red triangle shrinks to a focus-focus point.
b) Slice with $K=0$.}
\end{figure}
}

\section{Bifurcation diagram}

The energy-momentum map from $T^*SO(3)$ to $\R^3$ is given by $(l_z, L_3, H)$ where $L_3$
is given in terms of $\x$ and $\bm{l}$ as in Theorem~1. The bifurcation diagram of this integrable system is the set of critical values of the energy-momentum map.
Hence we are interested in the rank of $(X_{l_z}, X_{L_3}, X_H)$.
To determine where the rank drops we consider
\begin{equation} \label{eqn:combi}
    \alpha X_{L_3} + \beta X_{l_z} + \gamma X_H = 0 \,.
\end{equation}

\begin{lemma}
The rank 1 points of the energy-momentum map are given by two parabolas of sleeping tops
\begin{equation} \label{eqn:parabola}
    (l_z, L_3, H) = \left( m, \pm m,\frac{m^2}{2 I_1}( 1+ \delta) + V(\pm 1) \right) \,.
\end{equation}
The rank 2 points have a rational parametrisation determined by 
$
   \bm{l}(\beta, a_z) = \frac{1}{\beta} I_1 V'(a_z) \bm{a} + \beta \bm{e}_z
$
such that for $a_z \in [-1,1]$ and $\beta \in \mathbb{R}$ the critical values of the energy-momentum map are
\[
   l_z(\beta, a_z) = \frac{1}{\beta} I_1 V'(a_z) a_z + \beta, \quad
   L_3(\beta, a_z) = \frac{1}{\beta} I_1 V'(a_z) + \beta a_z, \quad
   H(\beta, a_z) = \frac{1}{2I_1}( \bm{l}(\beta, a_z)^2 + \delta L_3(\beta, a_z)^2) + V(a_z) \,.
\]
\end{lemma}
\begin{proof}
Notice that the last 3 components of $X_{V}$ can be written as
\[
  -\tfrac12 \x_+ \nabla_x V(a_z(x)) = -\bm{a} \times \nabla_a V(a_z) = -\bm{a} \times \bm{e}_z V'(a_z) \,.
\]
Using $
 \x_-^t \bm{e}_3 = \x_+^t R \bm{e}_3 = \x_+^t \bm{a}
$ in the flow of $L_3$, \eqref{eqn:combi} becomes
\[
\begin{pmatrix}\frac{1}{2}\bm{x}_{+}^{t}((\alpha+\gamma\delta L_{3})\bm{a}+\beta\bm{e}_{z}+\gamma\bm{l})\\
-(\gamma I_{1}V^{'}\bm{a}+\mu\bm{e}_{z}+\beta\bm{l})\times\bm{e}_{z}
\end{pmatrix}=\begin{pmatrix}0\\
0
\end{pmatrix}.
\]
This means critical points of the momentum map occur when 
\begin{equation}
\begin{aligned}(\alpha+\gamma\delta L_{3})\bm{a}+\beta\bm{e}_{z}+\gamma\bm{l} & =0\\
\gamma I_{1}V^{'}\bm{a}+\mu\bm{e}_{z}+\beta\bm{l} & =0.
\end{aligned}
\label{eq:2eq}
\end{equation}
for $\alpha$, $\beta$, $\gamma$ not all zero.
Hence the three vectors $\bm{a},\ \bm{l},\ \bm{e}_z$ are co-planar.
Since $\bm{a}$ and $\bm{e}_z$ never vanish, there is no rank 0 point. 
We have the following cases.
\begin{enumerate}
\item $\bm{a}\parallel\bm{e}_{z}$. 
This means $a_{x}=a_{y}=0$ and $a_{z}=\pm1$.
If $\bm{l} \neq 0$ and $\bm{l}$ not parallel to $\bm{e_{z}}$, then linear independence \eqref{eq:2eq} implies $\alpha=\beta=\gamma=0$.
Hence $\bm{l}\parallel\bm{e}_{z}\parallel\bm{a}$ (including $\bm{l} = 0$), and we can use $l_z=m$ as parameter and thus showed the parametrisation of the sleeping tops \eqref{eqn:parabola} in the Lemma.
Recall that these are the points where the torus action is singular.
All points along these parabolas have rank~1.
Parts of these parabolas may be isolated threads of focus-focus type, while others form the edges of the surface of elliptic-elliptic type.

\item $\bm{l}\parallel\bm{e}_{z}$. 
Set $\bm{l}=\lambda\bm{e}_{z}$.
If $\bm{a}\parallel\bm{e}_{z}$ then this gives the sleeping top solution again. If $\bm{a}$ is not parallel to $\bm{e}_{z}$ then by linear independence (\ref{eq:2eq}) implies 
$\alpha+\gamma\delta L_{3}=\beta+\gamma\lambda = 0$
 and 
 $I_{1}V^{'}=\mu+\beta\lambda = 0$.
If $\gamma=0$ then this forces the trivial solution $\alpha=\beta=0$ so $\gamma \not = 0$ and
$V^{'}=0$. Since $V^{'}=c_{1}+2c_{2}a_{z}$ this means $a_{z}=\frac{-c_{1}}{2c_{2}}\eqqcolon a_{z0}$.
Normalising $\gamma = -1$ gives $\lambda = \beta$
and hence $\bm{l}=\beta \bm{e}_{z}$,
$\bm{a}=(a_{x},a_{y},a_{z0})$,
$L_3 = \bm{l} \cdot \bm{a} = \beta a_{z0}$
and hence using $\beta = m$ as parameter gives
\[
(l_{z},L_{3},H)= \left(m,m a_{z0}, \frac{m^2}{2I_{1}}(1+\delta a_{z0}^{2})+V(a_{z0})\right).
\]
Since $|a_z| \le 1$ this parabola only exists when $|c_1| < 2 |c_2|$, while for $|c_1| = \pm 2 |c_2|$ it merges with 
the sleeping tops. Otherwise its vertex lies between the vertices of the parabolas of sleeping tops.
The vertex of this parabola $\beta = 0$ is an equilibrium point of $X_H$.

\item $\bm{l}\parallel\bm{a}$. This forces $\bm{l}=\lambda\bm{a}=L_{3}\bm{a}$.
If $\bm{a}\parallel\bm{e}_{z}$ then this gives the sleeping top solution again.
If $\bm{a}$ is not parallel to $\bm{e}_{z}$ then
linear independence (\ref{eq:2eq}) implies
$ \alpha + \gamma (\delta+1) L_3  = \beta = 0$
and
$ \gamma I_1 V' + \beta L_3 = \mu = 0$. Again $\gamma = 0$ gives the trivial solution, so we can normalise $\gamma = -1$ and find 
$\alpha = (1 + \delta) L_3$
and $V' = 0$, as in case 2.
Using $L_3=k$ as parameter gives
\[
(l_{z},L_{3},H)=
\left( k a_{z0},k,\frac{k^{2}}{2I_1}(\delta+1)+V(a_{z0})\right).
\]
Existence and limiting behaviour is as in case~2.
The vertex of this parabola coincides with that of case~2.

\item General case where no pair of vectors is parallel. If $\gamma=0$ then this gives $\alpha=\beta=0$ while if
$\beta=0$ then this gives case 2. We now assume $\beta\ne0$ and
$\gamma\ne0$. 
Eliminating $\bm{l}$ from \eqref{eq:2eq} and using linear independence gives
$\mu=\frac{\beta^{2}}{\gamma}$ and 
$\alpha + \gamma\delta L_{3}=\frac{\gamma^2}{\beta} I_1 V'$.
Using this to eliminate $L_3$ in \eqref{eq:2eq}
gives  
$-\bm{l}=\frac{\gamma}{\beta} I_{1}V^{'}\bm{a}+\frac{\beta}{\gamma}\bm{e}_{z}$.
Normalising $\gamma = -1$ computing $l_z = \bm{l} \cdot \bm{e}_3$, $L_3 = \bm{l} \cdot \bm{a}$, 
and $\bm{l}^2 = \bm{l} \cdot \bm{l}$ gives the result. Notice that $\beta$ is the angular velocity of the angle $\phi$ conjugate of $p_\phi$.
\end{enumerate}
\end{proof}

\begin{figure}
\begin{centering}
\includegraphics[width=6cm,height=6cm,keepaspectratio]{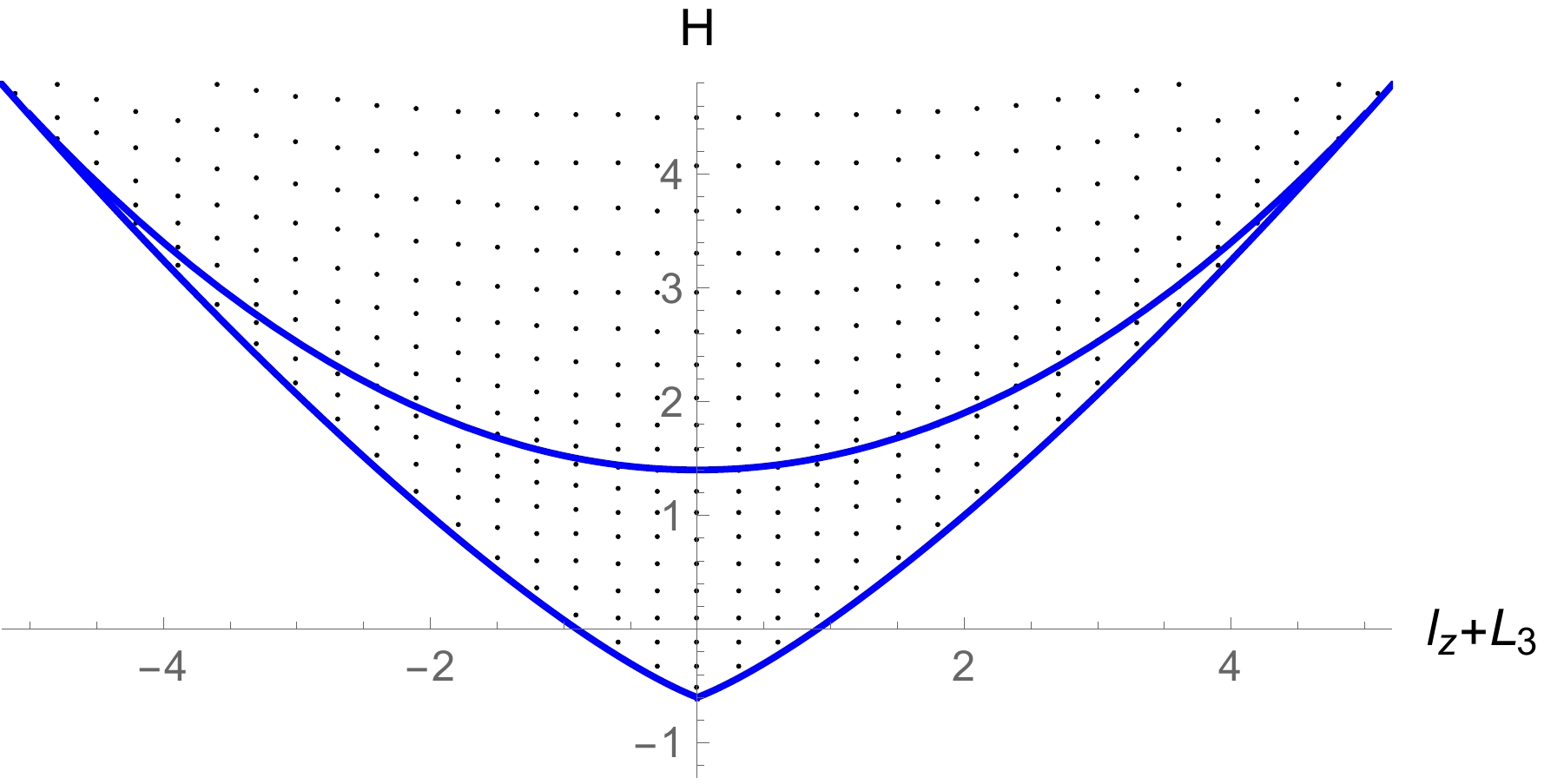}\includegraphics[width=6cm,height=6cm,keepaspectratio]{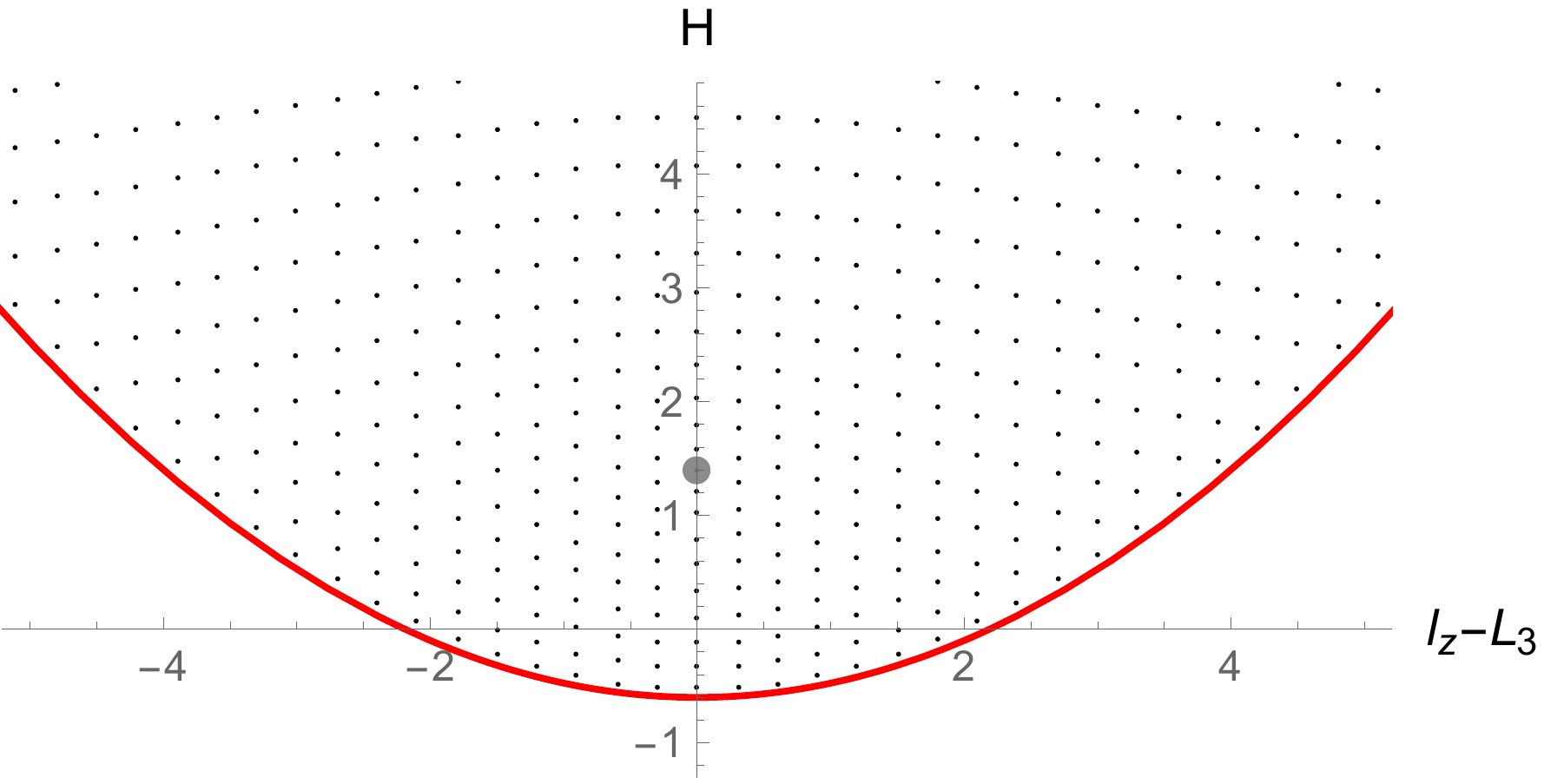}\\
\includegraphics[width=6cm,height=6cm,keepaspectratio]{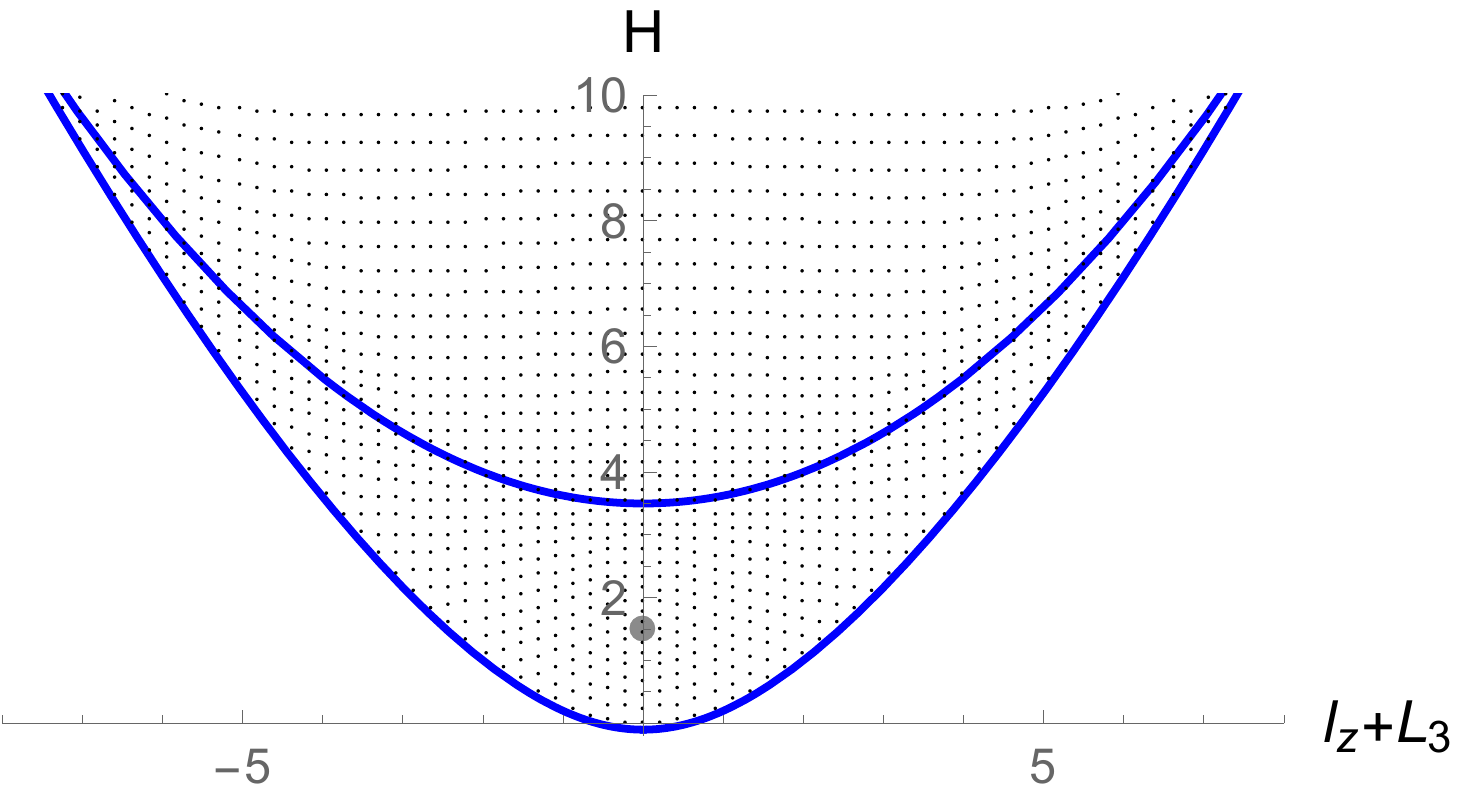}\includegraphics[width=6cm,height=6cm,keepaspectratio]{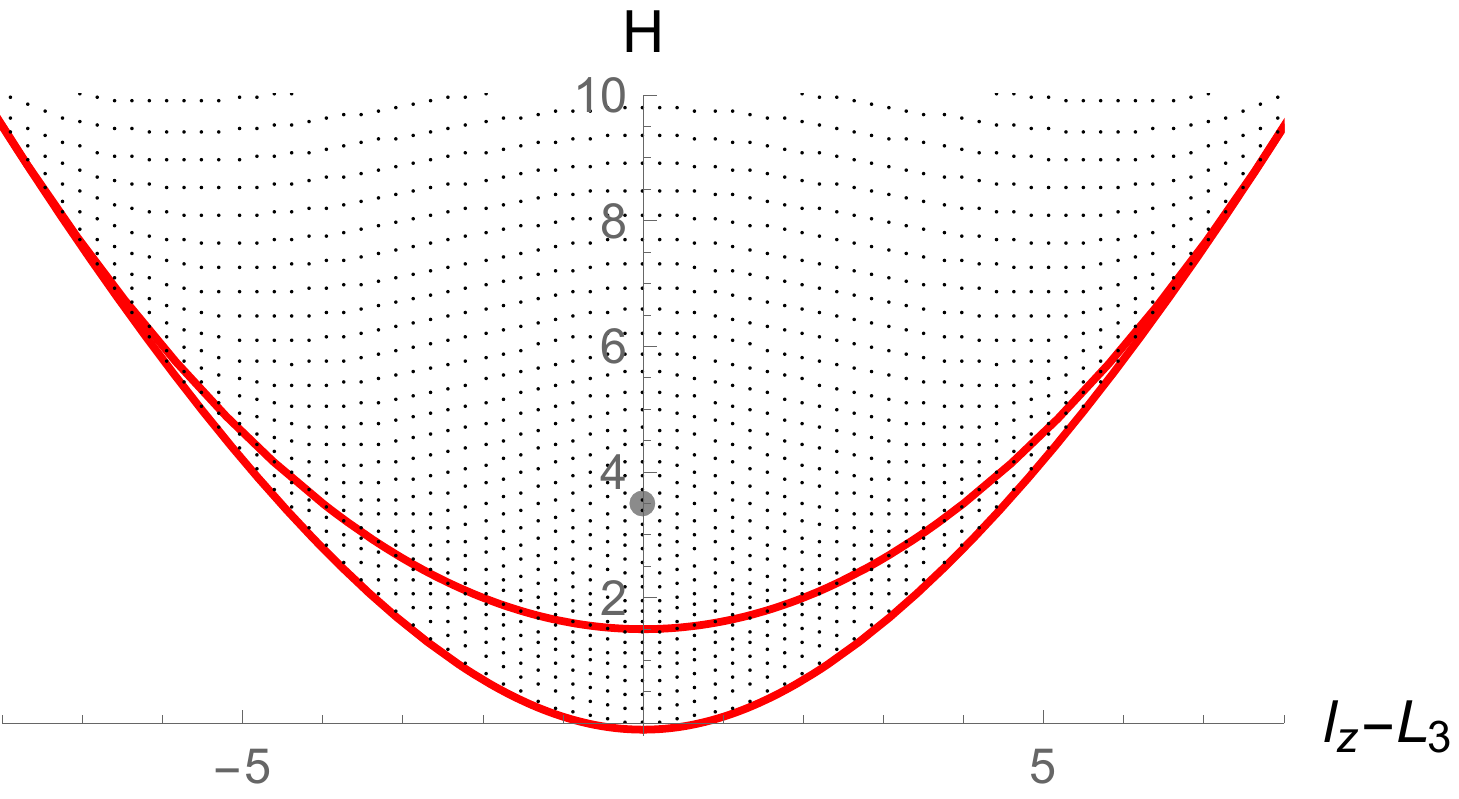}\\
\includegraphics[width=6cm,height=6cm,keepaspectratio]{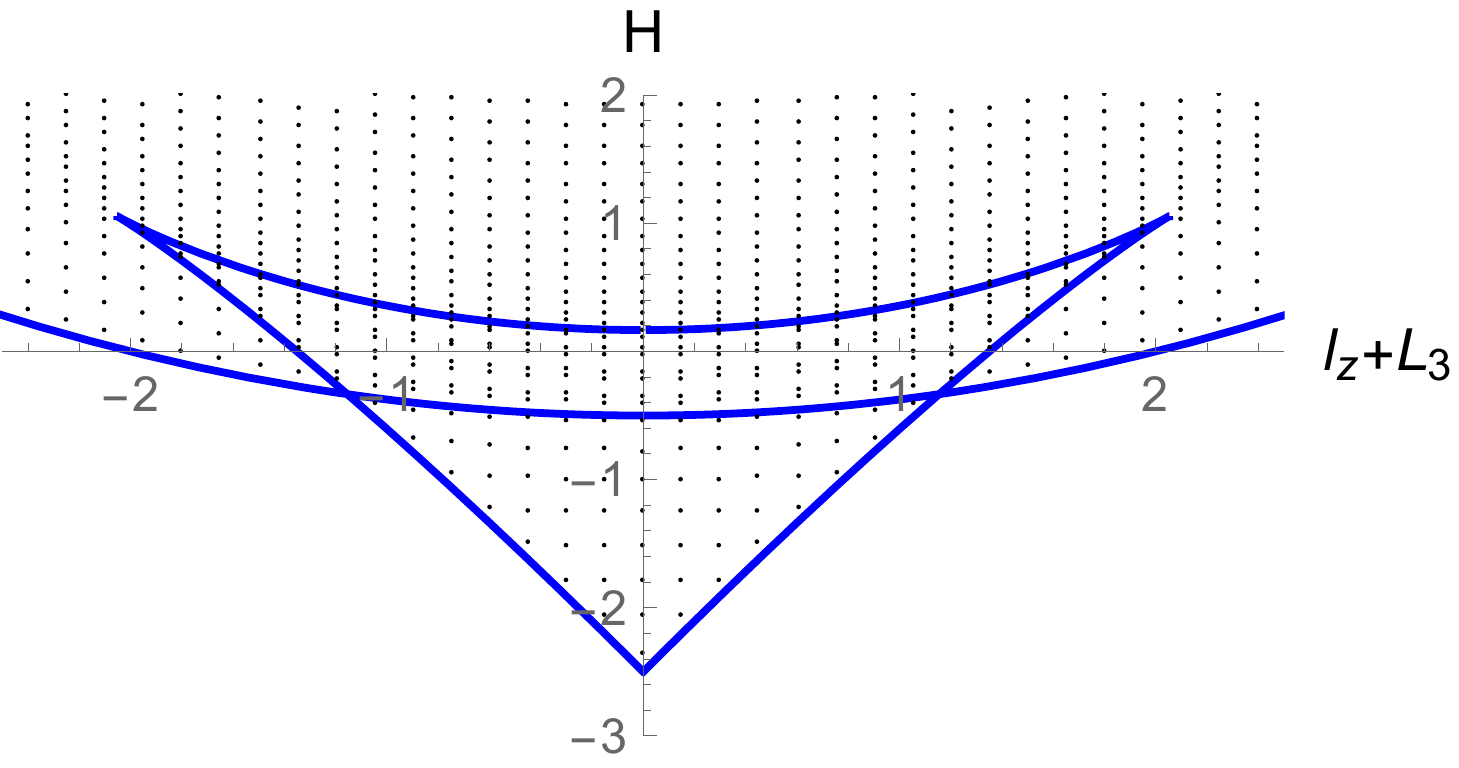}\includegraphics[width=6cm,height=6cm,keepaspectratio]{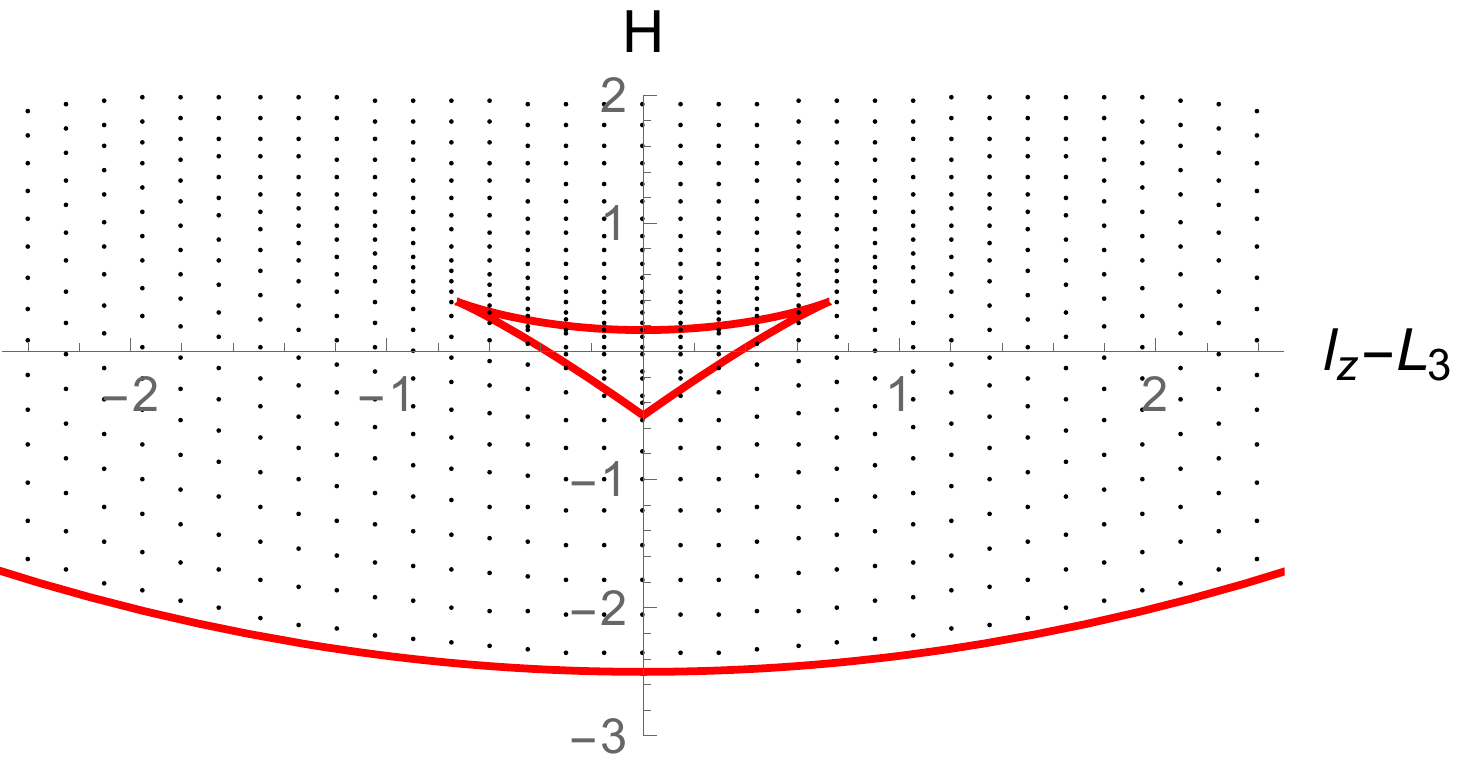}
\par\end{centering}
\caption{Slices along $l_{z}-L_{3}=0$ (blue) and $l_{z}+L_{3}=0$ (red) for Fig.~\ref{fig:3d examples} a) through
c) using $\hbar=(0.15,0.11,0.075)$ respectively. \label{fig:allslices}}

\end{figure}

Note that the two special cases 2 and 3 are parabolas that are embedded in the surface of critical values. Unlike the parabolas \eqref{eqn:parabola} the rank of these points is 2. 
The parabola from case~3 is in the closure of the generic case in the singular limit that $\beta \to 0$.


The rational parametrisation is also useful when using the Euler angles. 
Inserting the parametrisation into the condition for an equilibrium point $H_\theta = 0$ it is identically satisfied. To determine the stability of the equilibrium we evaluate the second derivative $H_{\theta\theta}$ on the rational parametrisation and find
\[
  H_{\theta\theta}(\beta, a_z) = \beta^2 I_1 - 2 a_z V'(a_z) + (1 - a_z^2) V''(a_z) + \frac{1}{\beta^2 I_1} V'(a_z)^2.
\]
Equating this to zero gives a relation between $\beta$ and $a_z$ which determines degenerate values in the bifurcation diagram. These are the cusp-shaped edges of the triangular tubes in Fig.~1c,d.
The most degenerate situation occurs when simultaneously the 2nd and the 3rd $\theta$-derivative of $H$ vanish. This occurs for the special parameter values 
$a_z = -c_1/(2 c_2)$, $\beta^2 = -c_1^2 / ( 8 c_2 I_1)$
and 
$a_z = -c_1/(4 c_2)$, $\beta^2 = c_1^2 / ( 2 c_2 I_1) - 2 c_2 / I_1$.
When these degenerate values for $a_z$ collide with $\pm 1$ then the degenerate points disappear and the topological structure of the bifurcation diagram changes.
This occurs for $c_1 = \pm 2 c_2$ and $c_1 = -4 c_2$. 
The plus sign yields imaginary $\beta$.
The sign of $c_1$ can be made positive by the original choice of body coordinate system.
Hence there are 4 topologically distinct cases illustrated in Fig.~1: 
$c_2/c_1 < -1/2$: triangular tube (1c);
$-1/2 < c_2/c_1 < -1/4$: triangular tube shrinking to a thread (1d);
$-1/4 < c_2/c_1 < 1/2$: one thread (1a);
$c_2/ c_1 > 1/2$ two threads (1a).

To understand the figures corresponding to these 4 cases it helps to consider how they bifurcate into each other. We stress again that we always consider $\delta = 0$, because adding the additional quadratic term in $L_3$ to the Hamiltonian deforms the bifurcation diagram, but does not essentially change it. 
Bifurcations similar to those found here have recently been described in \cite{SadovskiiZhilinskii07,Efstathiou19}, in particular also the related quantum monodromy in \cite{SadovskiiZhilinskii07}.
Let us start with the ordinary Lagrange top, $c_2 = 0$ \cite{Cushman85}. 
The bifurcation diagram for the harmonic Lagrange top is topologically the same for $-1/4 < c_2 /c_1 < 1/2$.
The outer surface is a bowl that has at least two corners when cut at constant energy. For high energy there are four corners, while for low energy only two. The transition is a supercritical Hopf bifurcation where the sleeping top becomes stable. A thread of isolated critical values detaches at that point of the surface. This thread is shown in blue in Fig.~1a. In Fig.~2 slices through the 3-dimensional bifurcation diagram are shown. The blue curves is a slice with $l_z - L_3 =0$ which contains the thread, while in the other slice $l_z + L_3 = 0$ the thread appears as a single isolated point. In these figures we also show the quantum spectrum, see the next section.
This situation persists for non-zero $c_2$ not too large. Changes occur for $c_2/c_1 = -1/4$ and for $c_2/c_1 = 1/2$. For $c_2/c_1 > 1/2$, a second thread emerges from the minimum of $H$, as show in Fig.~1b and Fig.~2b. For low energies, the outer surface has no corners at all. For intermediate energy as visible at the top of Fig.~1b, there are 2 corners above where the red thread is attached, but the blue thread is not yet attached and the outer surface is still smooth. For high energies, there are 4 corners.
A more dramatic change occurs when decreasing $c_2/c_1$ through $-1/4$.
All attachment points of the threads in the two cases discussed so far are supercritical Hopf bifurcations. When passing $-1/4$, the Hopf bifurcation turns into a subcritical Hopf bifurcation. The attachment point is replaced by a tube with triangular cross section that eventually contracts to a point and becomes a thread, as shown in Fig.~1d (zoomed in).
When decreasing $c_2/c_1$ further, the two subcritical Hopf bifurcation values collide
when $c_2/c_1 = -1/2$, and merge into a triangular tube shown in Fig.~1c.
In this figure, the bounding box is chosen such that it cuts away parts of the surface facing the viewer so that the triangular tube becomes visible.
The 0-slices are shown in Fig.~2c. The two bottom surfaces of the tube correspond to elliptic 2-tori, while the top surface of the tube corresponds to hyperbolic 2-tori. The top surface joins the bottom surfaces along a line of cusps where $H_{\theta\theta} = 0$. The merging of the triangular tube with the outer surface is illustrated in additional slices in Fig.~3.


\rem{
Another interpretation of the set of critical values of the energy momentum map is as follows.
The differential $p_\theta d\theta$ where $p_\theta$ is determined by $H(\theta, p_\theta) = h$ is a differential on an elliptic curve. The Lemma gives a rational parametrisation of the part of the discriminant surface of this curve that has $|a_z| \le 1$.
}



\begin{figure}
\includegraphics[width=4cm,height=4cm,keepaspectratio]{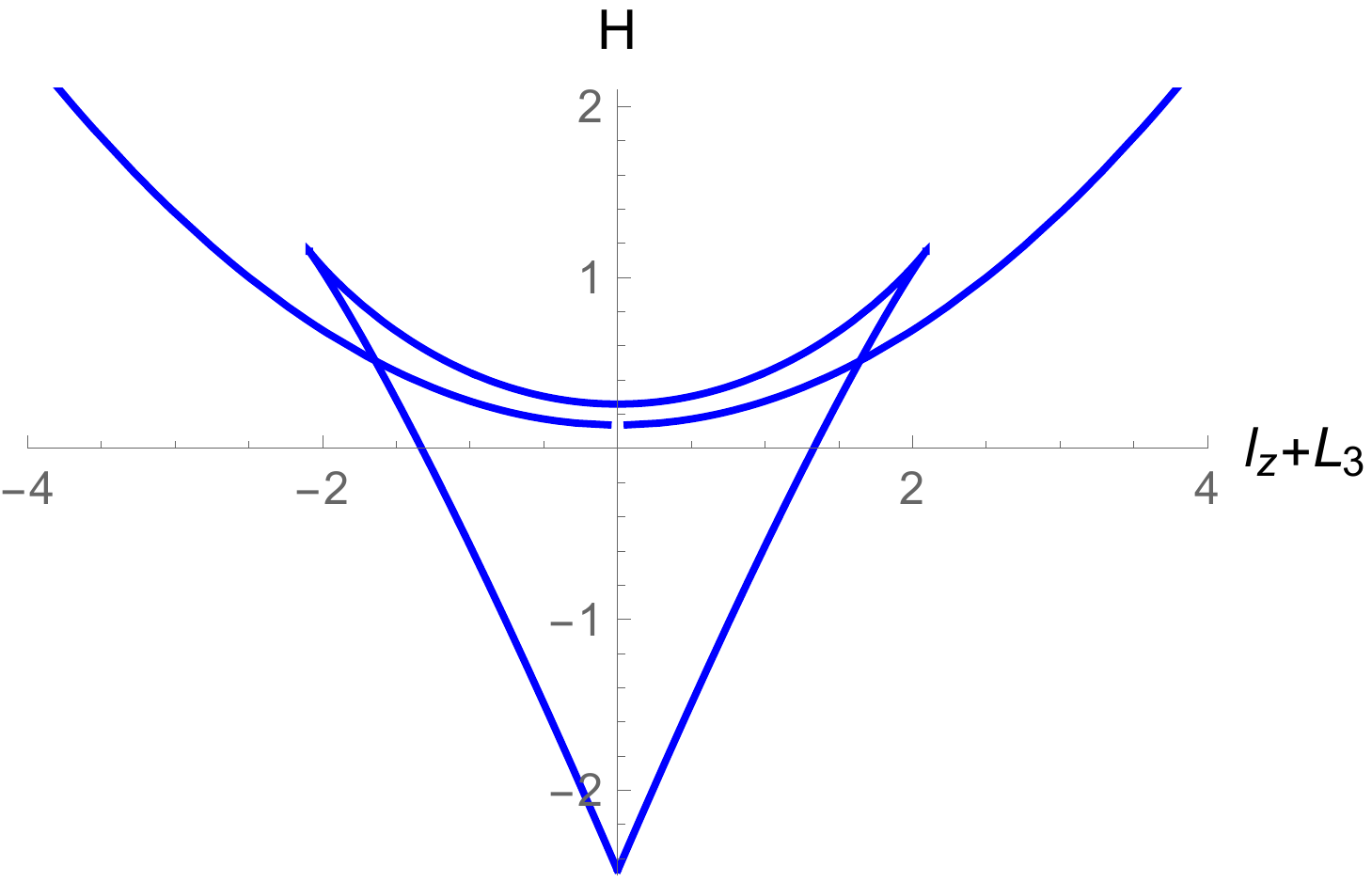}\includegraphics[width=4cm,height=4cm,keepaspectratio]{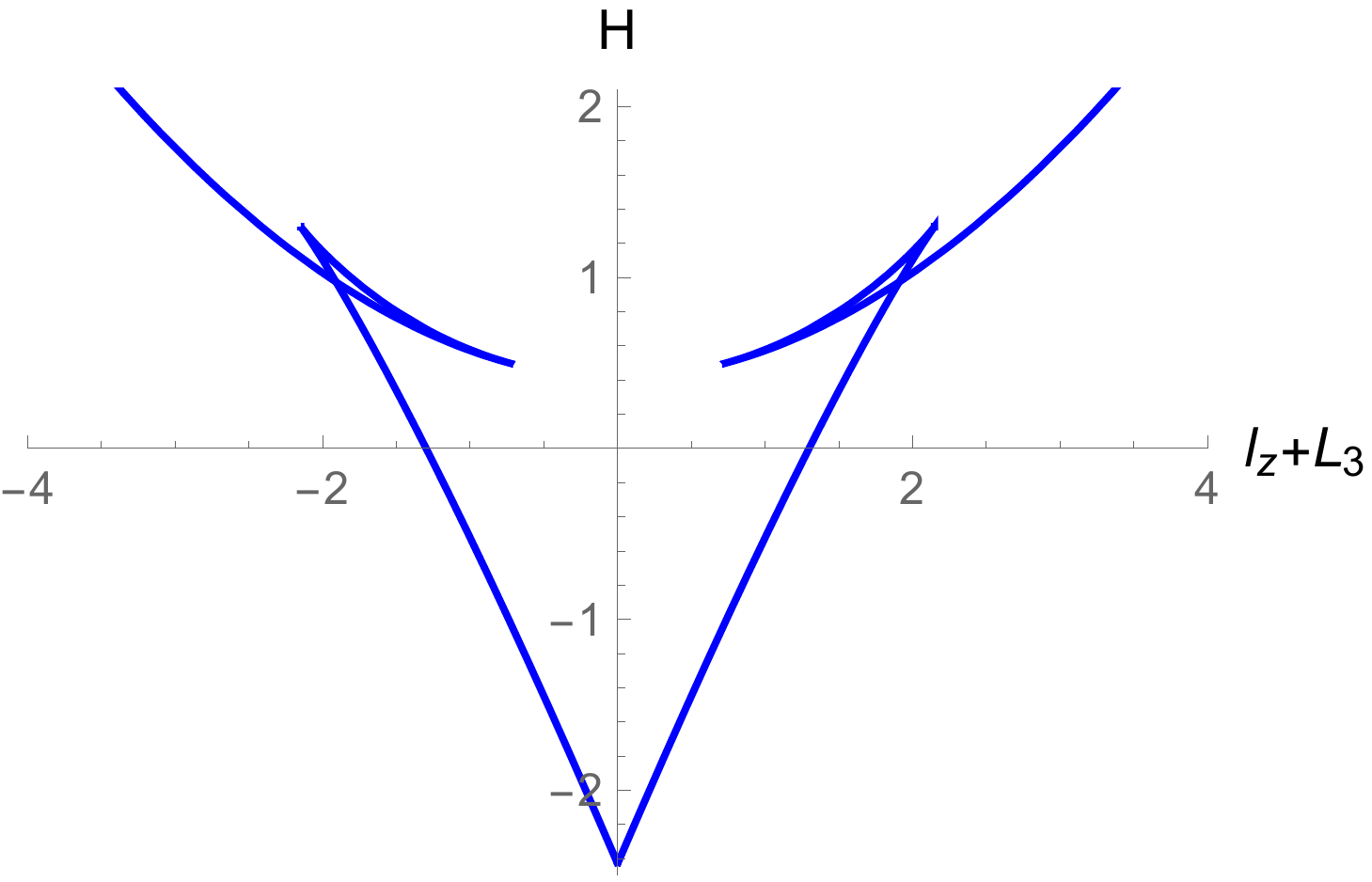}\includegraphics[width=4cm,height=4cm,keepaspectratio]{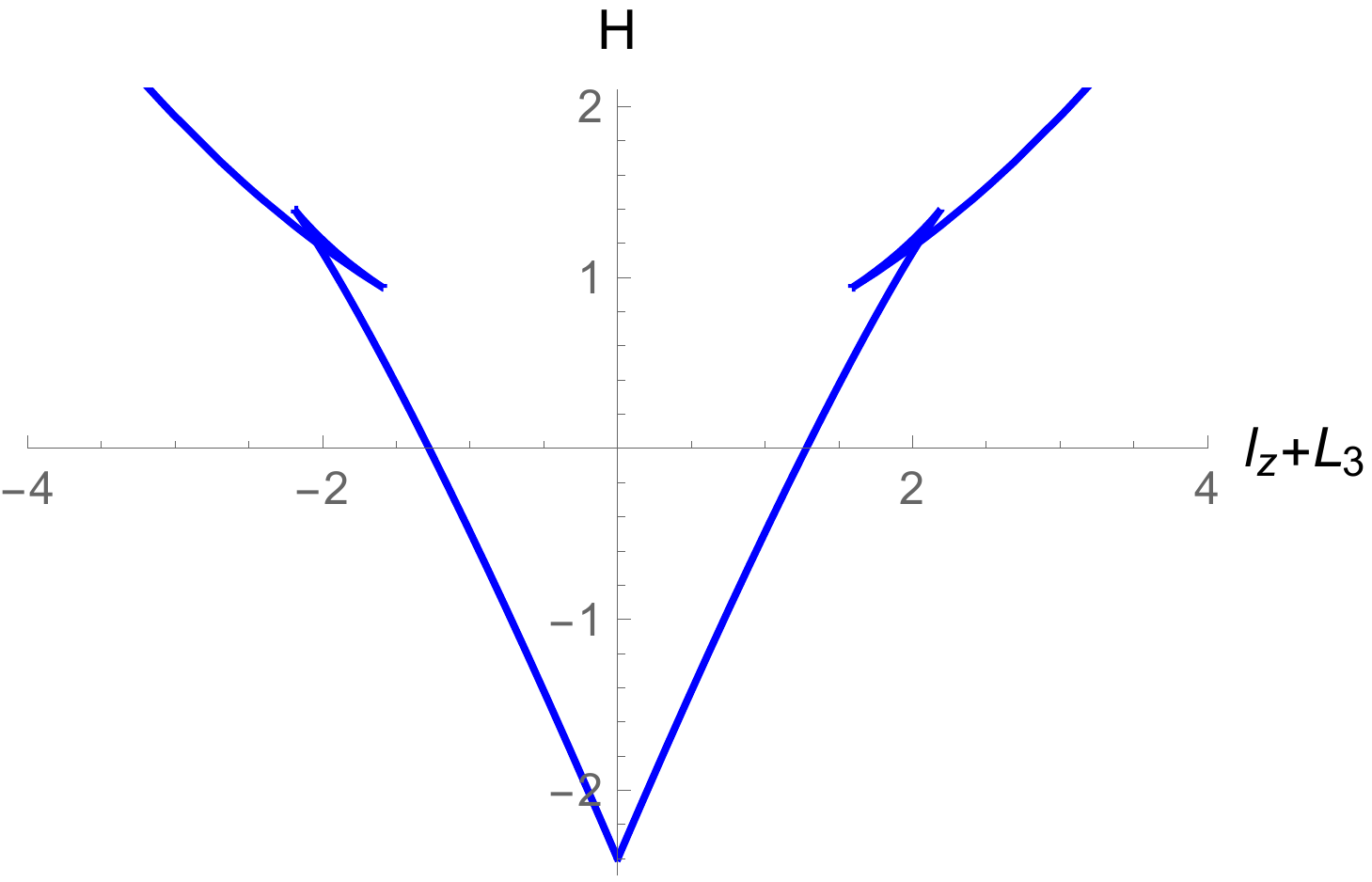}\includegraphics[width=4cm,height=4cm,keepaspectratio]{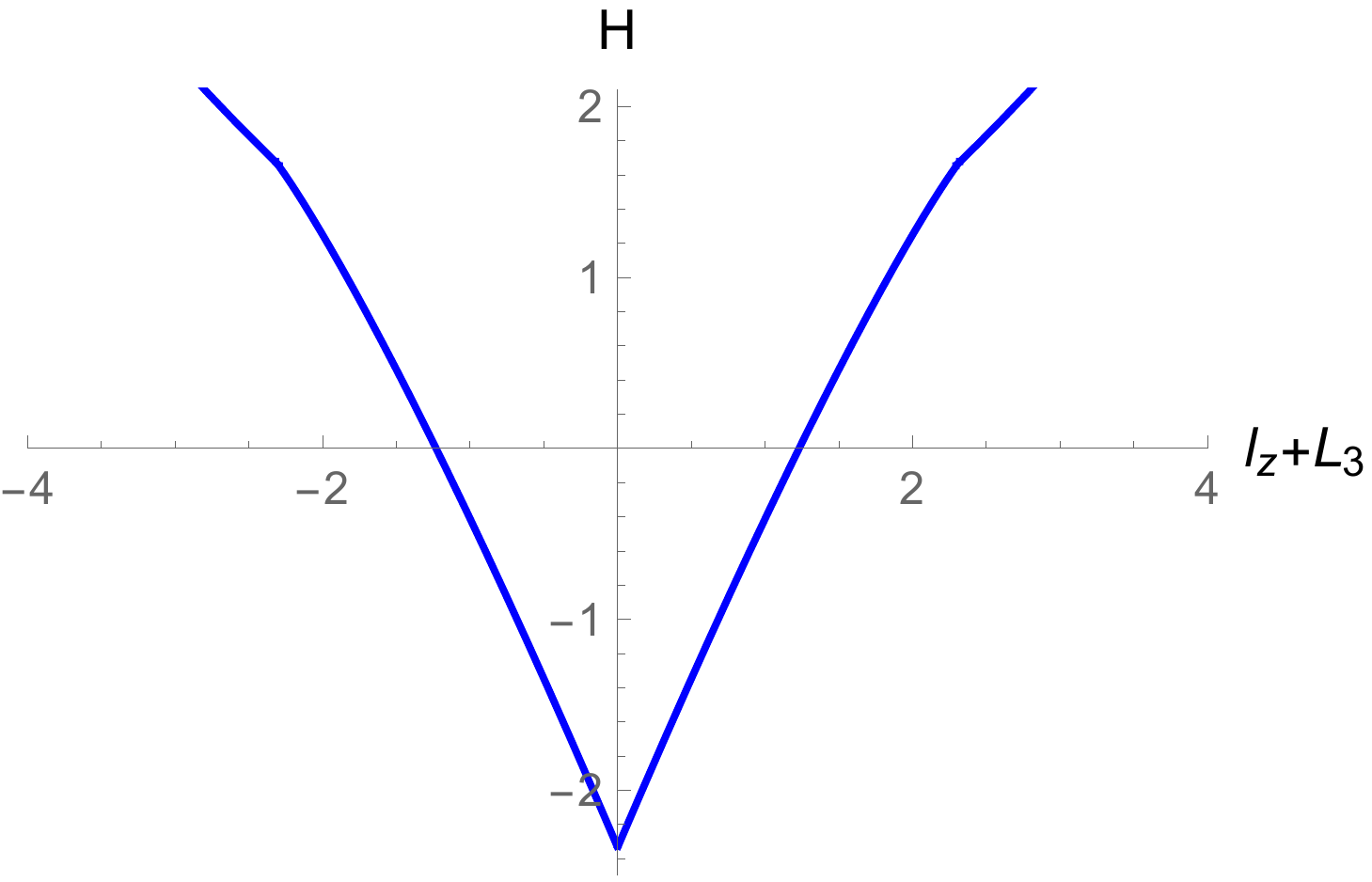}\\
\includegraphics[width=4cm,height=4cm,keepaspectratio]{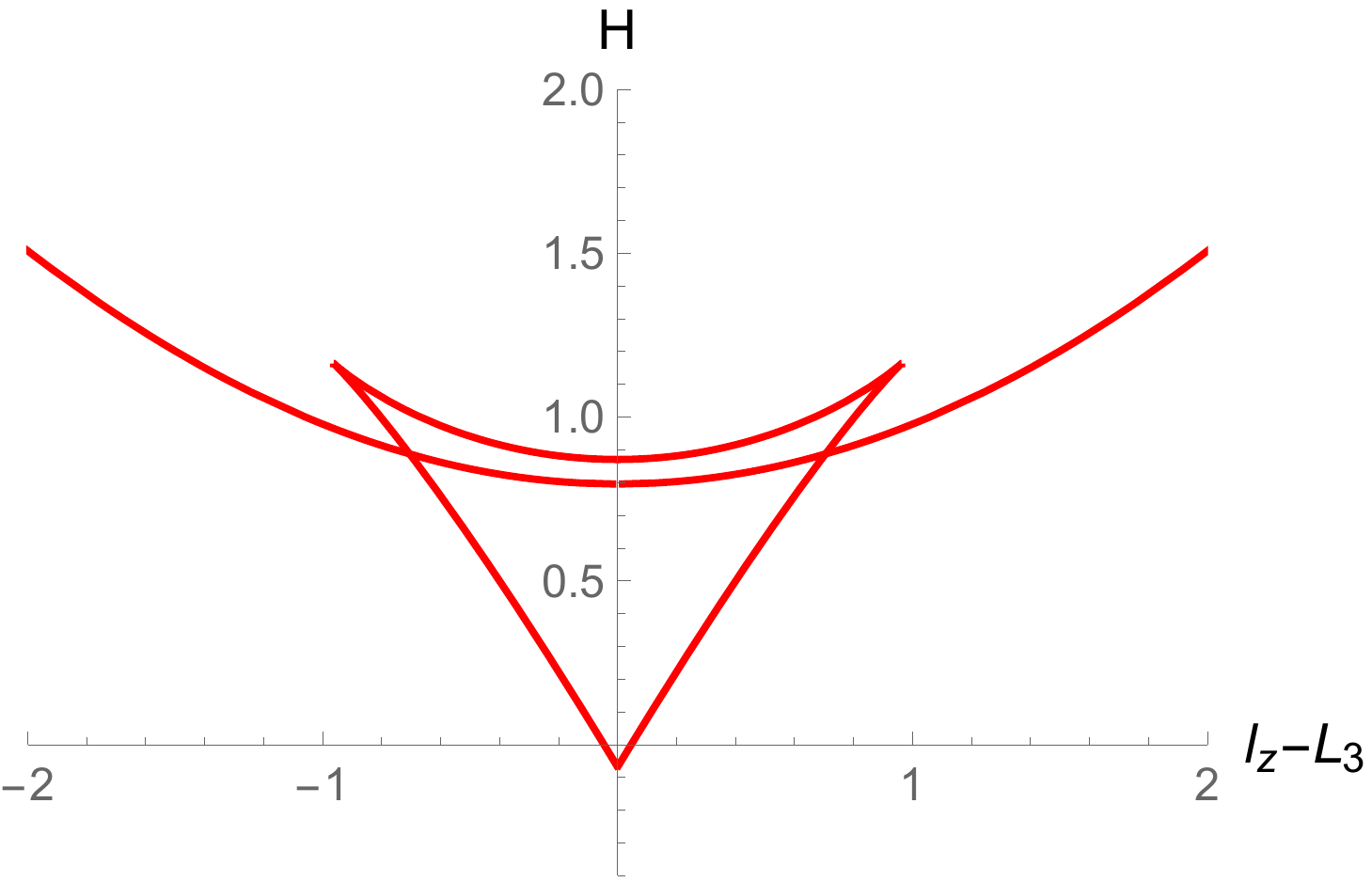}\includegraphics[width=4cm,height=4cm,keepaspectratio]{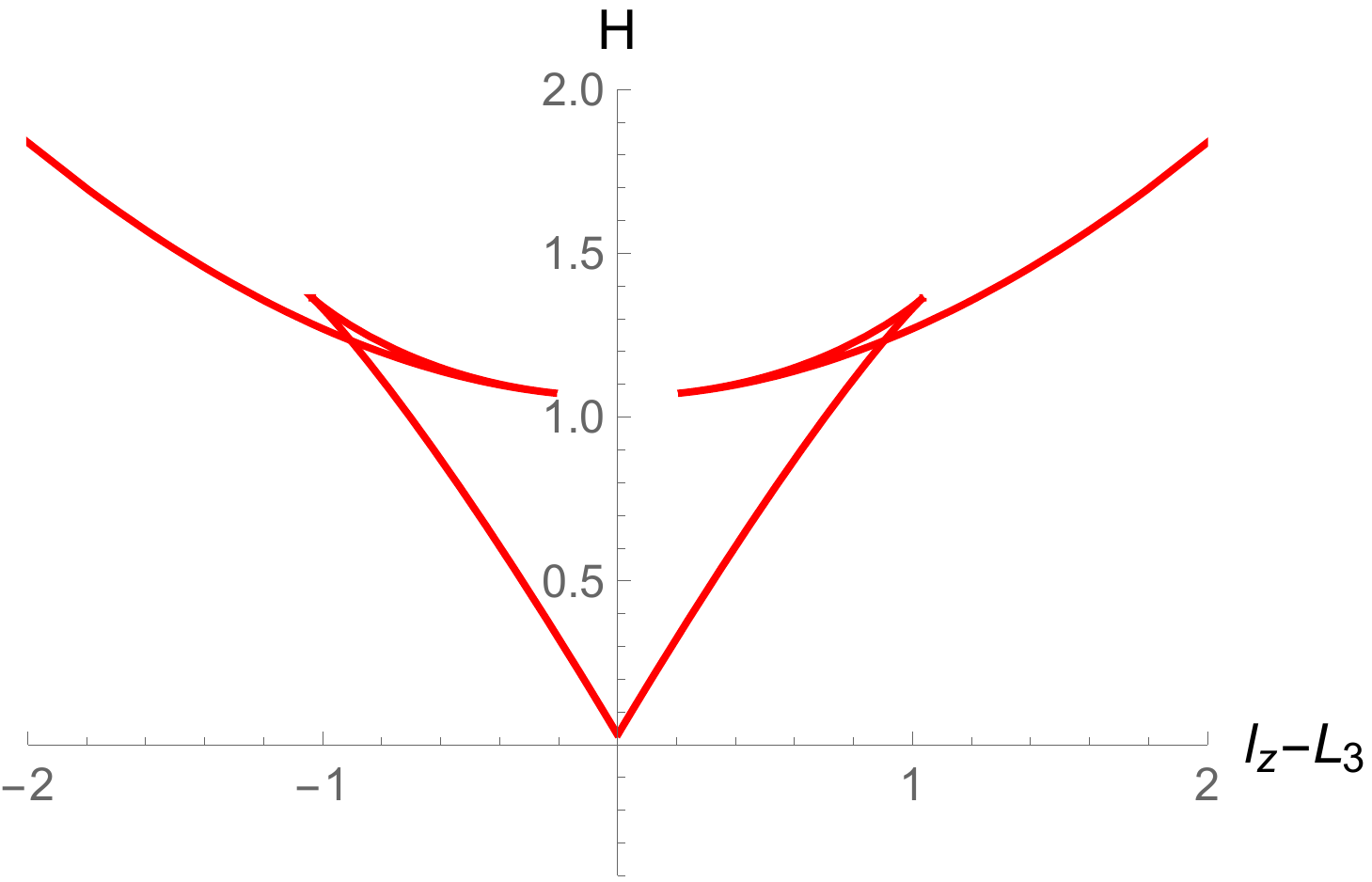}\includegraphics[width=4cm,height=4cm,keepaspectratio]{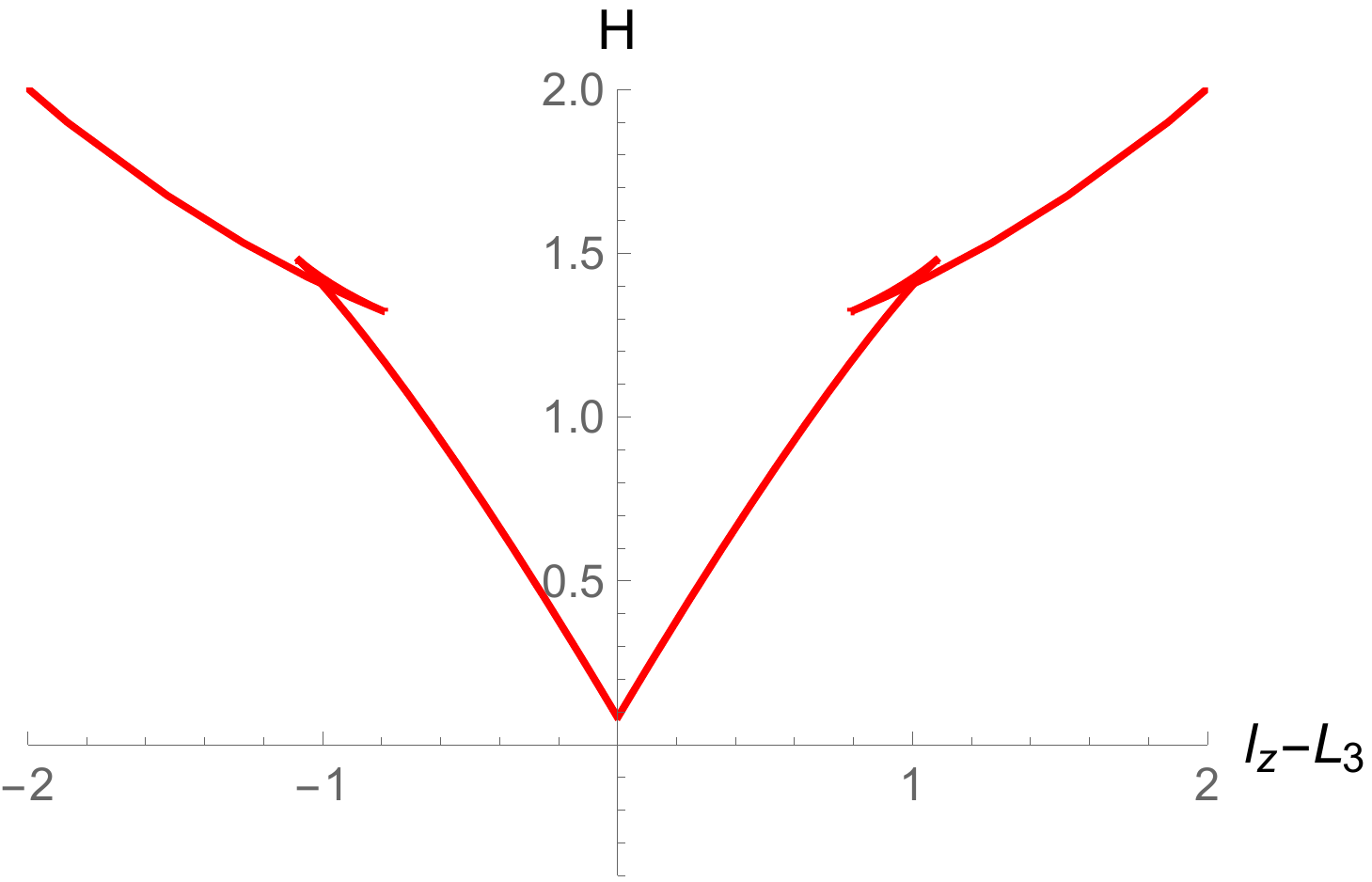}\includegraphics[width=4cm,height=4cm,keepaspectratio]{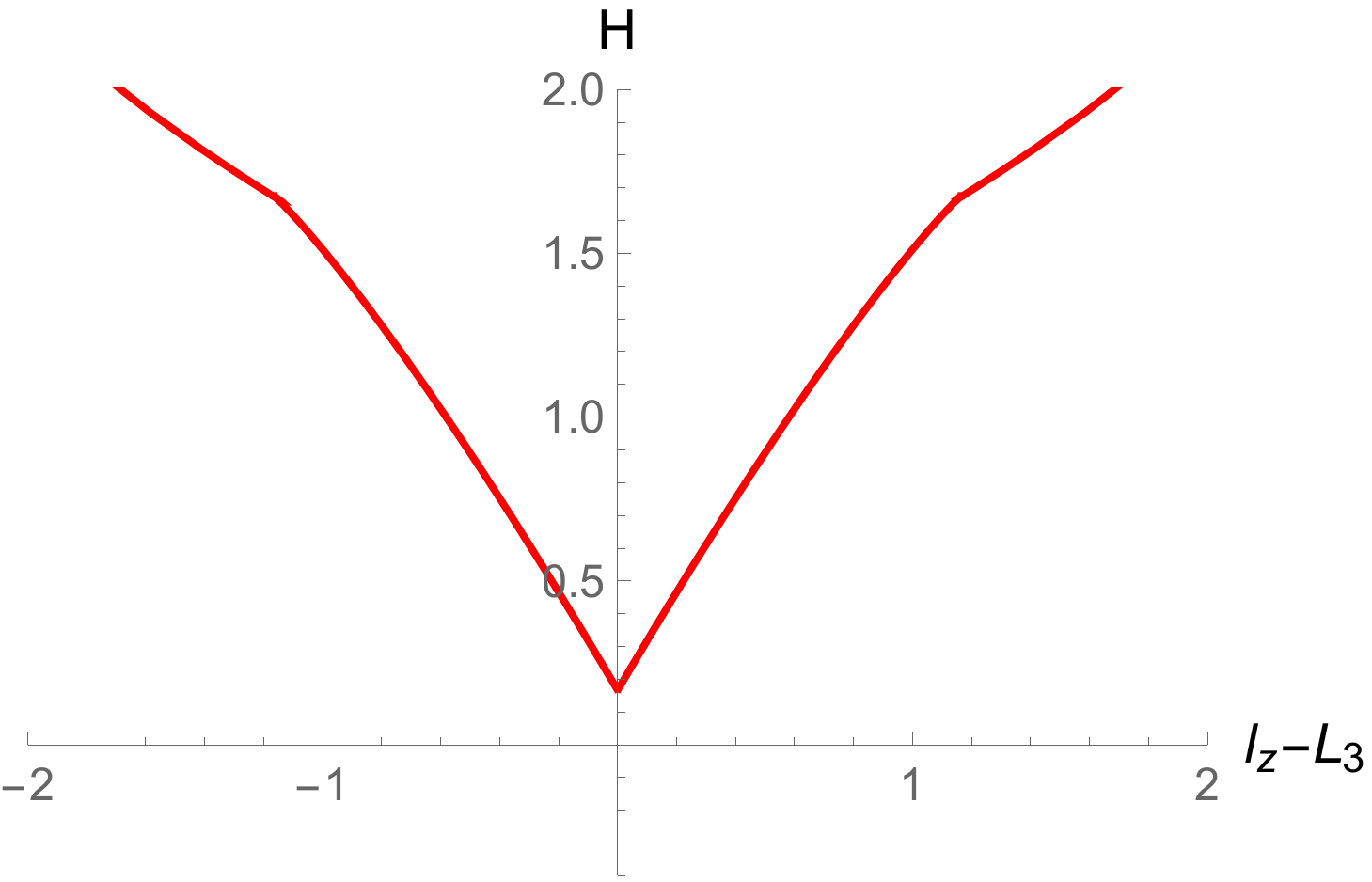}

\caption{Slices of constant $l_{z}-L_{3}$ and $l_{z}+L_{3}$ near the most degenerate values.
Top: slices with constant $l_z - L_3=(0.488,0.755,0.888,1.15).$ 
Bottom: slices with constant $l_z + L_3=(1.96,2.06,2.16,2.31).$ Parameters are the same as those in Fig. 1c. \label{fig:Pinch pic}}

\end{figure}



\rem{
\begin{figure}[H]
\begin{centering}
\includegraphics[width=8cm,height=8cm,keepaspectratio]{limit2d}\includegraphics[width=8cm,height=8cm,keepaspectratio]{\string"limit2d red\string".pdf}
\par\end{centering}
\caption{a) An oblate momentum map appears in the cross section of the limiting
case when $M=0$. b) Similarly, a prolate section appears in the $K=0$
section. WHAT ARE the PARAMETERS? \label{fig: slicespec}}
\end{figure}
}

\section{Quantum Mechanics of the Harmonic Lagrange Top}

The quantisation of the rigid body is textbook material, see, e.g., \cite[\textsection{103}]{landau13quantum}. The global action variables $l_z$ and $L_3$ become operators $\hat l_z = -i \partial / \partial \phi$ and $\hat L_3 = -i  \partial / \partial \psi$, measured in units of $\hbar$. We denote the corresponding integer eigenvalues by $m$ and $k$ such that $\hat l_z \Psi = m \Psi$ and $\hat L_3 \Psi = k \Psi$ for a wave function $\Psi$.

The quantum mechanical harmonic Lagrange top has the Hamiltonian
operator
\begin{equation}
\hat{H}=
\frac{1}{2 I_1}\left(\hat{\bm{l}}^{2}+\delta \hat{L}_3^{2}\right)  
+c_{1}\cos\theta+c_{2}\cos^{2}\theta\label{eq:hlt}
\end{equation}
where $\hat{\bm{l}}$ is the total angular momentum operator. 
Explicitly the first part of the Hamiltonian operator is found as the Laplace-Beltrami operator of the metric $T_{round}$ of the spherical top, hence
\[
\hat{ \bm{l}}^2 = -\frac{1}{\sin\theta} \partial_\theta( \sin\theta \, \partial_\theta) + \frac{1}{\sin^2\theta}( m^2 + k^2 -2 mk \cos\theta ) \,,
\]
where we have already replaced the operator $\hat{l}_z$ and $\hat{L}_3$ by their respective eigenvalues.
\rem{
The Laplace-Beltrami operator for the metric $g_{ij}$ is
\[
   \Delta f = \frac{1}{\sqrt{|g|}} \sum_{i,j} \partial_i \left( \sqrt{|g|} g^{ij} \partial_j f\right) 
\]
where $g^{ij}$ denotes the inverse of the metric and $|g|$ its determinant. } 
The equation $\hat{\bm{l}}^2 f = j(j+1) f$ is a self-adjoint form of the hypergeometric equation. 
Setting the eigenvalue of $\hat{\bm{l}}^2$ to $j(j+1)$ for positive integer $j$, solutions are given by $\sin^{|m+k|}\tfrac{\theta}{2} \cos^{|m-k|}\tfrac{\theta}{2} P_{j-\max(|k|,|m|)}^{|m+k|,|m-k|}(\cos\theta)$ where $P_n^{m_1,m_2}$ are the Jacobi polynomials. Up to normalisation and phase factors these are the Wigner-$D$ functions.
The equation has regular singular points at $\theta = 0, \pi$ with indices $\pm(m-k)$ and $\pm(m+k)$, respectively. 
Note that the global quantum numbers $m$ and $k$ appear as indices of regular singular points.

Adding the potential terms, and transforming to $z = \arccos\theta$ brings us to the following observation.
\begin{theorem}
The quantisation of the harmonic Lagrange top leads to the most general confluent Heun equation (aka generalised spheroidal wave equation) which has the self-adjoint form
\begin{equation} \label{eqn:cHeun}
    -\partial_z( (1 - z^2) \partial_z + \frac{k^2 + m^2 - 2 k m z}{1 - z^2} + \tilde c_1 z + \tilde c_2 z^2 - \lambda = 0
\end{equation}
where $z = \cos\theta$ and $\lambda$ is the spectral parameter related to the energy eigenvalue $E$ of the Hamiltonian by $\lambda = 2 I_1 E / \hbar^2 - \delta k^2$, 
$\tilde c_1 = c_1 2 I_1/\hbar^2$, 
$\tilde c_2 = c_2 2 I_1/\hbar^2$, 
\end{theorem}
In the form \eqref{eqn:cHeun} the indices at $z = \pm 1$ are $\pm(m-k)/2$ and $\pm(m+k)/2$. 
This equation has an irregular singular point at infinity, which is obtained by the confluence of two regular singular points of the Heun equation. The Heun equation is the most general Fuchsian equation with 4 regular singular points. The Heun equation (after normalisation by M\"obius transformations) has 6 parameters, 1 position of a pole, 4 indices, and the so called accessory parameter. The pole position is used for the confluence, after which only two regular singular points remain. Hence 2 indices remain as parameters (related to $k$ and $m$). Two additional parameters describe the behaviour near the irregular singular point, and the accessory parameter remains, so that there is a total of 5 parameters. 

To transform into the standard form of the confluent Heun equation, see, e.g., \cite{DLMF}, first shift to the standard poles by $ z \to ( z+1)/2$, and then scale the dependent variable with $\exp(2\sqrt{\tilde c_2})z^{|m+k|}(z-1)^{|m-k|}$.

When considering the confluent Heun equation, the usual reference to its application in physics is to Teukolsky's master equation \cite{teukolsky1973perturbations}, which appears in the perturbation theory around a rotating  black hole, i.e.~the Kerr metric. However, that equation only has 4 parameters, and one index-parameter is more restricted because it represents the spin of a particle. In this context, eigenvalues $\lambda$ of the equation have been computed using expansion in Jacobi polynomials in \cite{Crossman77}. Their results are not applicable to our case because their equation only has 4 parameters. To compute the spectrum in our case we generalise the papers 
\cite{Shirley63stark}, \cite{Hajnal91stark} which treat the case of a symmetric molecule (i.e.~top) in an electric field, hence the Lagrange top (without the harmonic field). To extend their method, which is also an expansion in Jacobi polynomials (or rather the related Wigner $D$-functions), we need to compute the matrix elements of $\cos^2\theta$. This leads to our final result.

\begin{theorem}
The spectrum of the harmonic Lagrange top \eqref{eq:hlt} which is equivalent to the most general confluent Heun equation  \eqref{eqn:cHeun} can be computed from a penta-diagonal symmetric matrix
\begin{equation}
\hat H = \hat H_0 + c_1 \hat H_1 + c_2 \hat H_1^2\,.
\end{equation}
For given fixed $m,k$ the operator $\hat H_0$ is the diagonal representation of the Hamiltonian without potential and $H_1$ is the tri-diagonal representation of $\cos\theta$ in terms of Wigner-$D$ basis functions.
\rem{
\begin{equation}
\hat{H}=\begin{pmatrix}C_{0} & B_{1} & A_{2}\\
B_{1} & C_{1} & B_{2} & A_{3}\\
A_{2} & B_{2} & C_{2} & B_{3} & A_{4}\\
 & A_{3} & B_{3} & C_{3} & B_{4} & A_{5}\\
 &  &  &  &  &  & \ddots
\end{pmatrix}\label{eq:Hmat}
\end{equation}
which is the representation of the Hamiltonian operator of the harmonic Lagrange top in the basis of Wigner-$D$ functions. The entries of the matrix are given by
\[
C_i = ..., \quad
B_i = ..., \quad
A_i = ...
\]
\footnote{the notation $C,B,A$ seems unusual, wouldn't $A,B,C$ be more natural?}
where $C_{a_{1}b_{1}a_{2}b_{2}}^{AB}$ are the Clebsch-Gordon coefficients.
} 
\end{theorem}
\begin{proof}
The formulas for $\hat H_0$ and $\hat H_1$ are given in \cite{Shirley63stark}. We repeat them here for convenience.
The diagonal entries of $\hat H_0$ are $\frac{\hbar^2}{2I_1}( j(j+1) + \delta k^2)$. The diagonal entries of $\hat H_1$ are
$a_j = -k m / ( j ( j+1))$ and the off-diagonal entries are 
$b_j = -\sqrt{ ( j^2 - k^2)(j^2 - m^2)/(j^2 (4 j^2 - 1)) }$.
The first entries in the matrix representing the operators have  $j = \max(m,k)$. Note that for $m=k=0$ the diverging terms in $b_j$ cancel and $b_0$ is defined.
It is easy to compute the matrix elements of $\cos^2\theta$. This can be done by noticing that  $D_{2,0,0} = \tfrac32 \cos^2\theta - \tfrac12$. The matrix representation of $D_{2,0,0}$ can be expressed in terms of Clebsch-Gordan coefficients. However, it is more efficient to use the fact that since $\hat H_1$ represents $\cos\theta$ the matrix $\hat H_1^2$ represents $\cos^2\theta$. 
So instead of computing matrix elements of $\cos^2\theta$ from scratch in terms of Clebsch-Gordan coefficients we can simply compute the square of the matrix representation of $\hat H_1$. In particular the entries in the 2nd off-diagonal are  given by products $b_{j-1} b_j$.
\end{proof}

The numerical convergence of these expressions is good, and the spectra displayed in Figure~2 were computing from these matrices truncated at twice the maximal needed quantum number $j$.
Even though the term $\delta L_3^2$ in the Hamiltonian is important for the classical dynamics,
its effect on the quantum spectrum is rather trivial, it simply adds $\delta k^2$. It does change the spectrum, but the change is simple, and for this reason in the figures we restricted attention to $\delta = 0$, the spherical top. Moreover, from the point of view of the computation of the spectrum of the general confluent Heun equation the term $\delta k^2$ is irrelevant. 

Why is there a correspondence between the harmonic Lagrange top and the confluent Heun equation? This question may not have a definite answer, but it is suggestive that the harmonic potential is the most general potential for which the classical dynamics can be linearised using the Jacobian of an elliptic curve. This fact appears to be related to the fact that the corresponding quantum system is described by the confluent Heun equation. After adding higher order terms to the potential, the system remains integrable and separable in the same way, but the classical dynamics will involve hyperelliptic curves, and the quantum system will be described by higher order confluent Fuchsian equations. It would be interesting to make this observation more precise.




\rem{
\item Redefine the ``Harmonic Lagrange top'' which is the canonical Lagrange
top with an added harmonic potential with arbitrary coefficient:
\begin{equation}
\mathcal{H}=\frac{\left\langle \bm{\mu},\bm{\mu}\right\rangle }{2I_{1}}-\frac{I_{3}-I_{1}}{I_{1}I_{3}}\frac{a^{2}}{2}+\chi\zeta_{3}+k\zeta_{3}^{2}.\label{eq:H harmonc}
\end{equation}
\item Using the symplectic coordinates defined in the previous section we
rewrite (\ref{eq:H harmonc}) as 
\[
\tilde{\mathcal{H}}\coloneqq I_{1}\mathcal{H}=p^{2}(1-q^{2})+\frac{a^{2}+m^{2}-2amq}{1-q^{2}}+c_{0}+c_{1}q+c_{2}q^{2}
\]
 where $c_{0}=2(\frac{-a^{2}}{2I_{3}}(I_{3}-I_{1})-\frac{a^{2}}{2}),c_{1}=2\chi I_{1},c_{2}=2\chi k$. 

\item Using Weyl quantisation $p\to\frac{1}{i}\frac{\partial}{\partial q}$
we obtain 
\[
\hat{\mathcal{H}}\psi=\left(-\frac{\partial}{\partial q}(1-q^{2})\frac{\partial}{\partial q}+\frac{a^{2}+m^{2}-2amq}{1-q^{2}}+c_{0}+c_{1}q+c_{2}q^{2}\right)\psi
\]
} 
 
 \rem{
\item Shift poles by change of dependent variable $q\to\frac{q+1}{2}$.
Alter roots of indicial equation with $\psi=\exp(\frac{2\sqrt{c_{2}}}{\hbar})t^{\left|\frac{a+m}{2\hbar}\right|}(t-1)^{\left|\frac{a-m}{2\hbar}\right|}W$.
\item New ODE:
\[
W^{''}+\left(\frac{\left|\frac{a-m}{\hbar}\right|+1}{t-1}+\frac{\left|\frac{a+m}{\hbar}\right|+1}{t}+\frac{4\sqrt{c_{2}}}{\hbar}\right)W^{'}+\frac{Xt+Y}{2t(t-1)\hbar^{2}}=0
\]
 where $X=4\sqrt{c_{2}}\hbar\left|\frac{a-m}{\hbar}\right|+4\sqrt{c_{2}}\hbar\left|\frac{a+m}{\hbar}\right|+8\sqrt{c_{2}}\hbar-4c_{1}$
and $Y=a^{2}-2\text{c0}+m^{2}+2\lambda+\hbar^{2}\left|\frac{a-m}{\hbar}\right|+\hbar^{2}\left|\frac{a+m}{\hbar}\right|+\hbar^{2}\left|\frac{a-m}{\hbar}\right|\left|\frac{a+m}{\hbar}\right|+2c_{1}-4\sqrt{c_{2}}\hbar\left|\frac{a+m}{\hbar}\right|-4\sqrt{c_{2}}\hbar-2c_{2}.$ 
\item Roots of indicial equation are $0,-\left|\frac{a+m}{\hbar}\right|$
at $t=0$ and $0,\left|\frac{a-m}{\hbar}\right|$ at $t=1$. 
} 

\rem{
\item Three term recurrence:
\[
Ac_{k+1}+Bc_{k}+Cc_{k-1}=0
\]
\[
\begin{aligned}A & =-2(k+1)\hbar^{2}\left(\left|\frac{a+m}{\hbar}\right|+k+1\right)\\
B & =a^{2}+|a+m|\left(-4\sqrt{c_{2}}+2k\hbar+\hbar\right)+|a-m|(|a+m|+2k\hbar+\hbar)-2c_{0}+2c_{1}-8\sqrt{c_{2}}k\hbar-4\sqrt{c_{2}}\hbar-2c_{2}+2k^{2}\hbar^{2}+2k\hbar^{2}+2\lambda+m^{2}\\
C & =4\sqrt{c_{2}}\hbar\left|\frac{a-m}{\hbar}\right|+4\sqrt{c_{2}}\hbar\left|\frac{a+m}{\hbar}\right|-4c_{1}+8\sqrt{c_{2}}k\hbar
\end{aligned}
\]
} 

\rem{
\section{The Momentum Map and Critical Values}

We now reduce to one degree of freedom by quotienting by $\mu_{3}$.
Invariants of this symmetry are 
\[
(z_{1},z_{2},z_{3},z_{4},z_{5},z_{6})\coloneqq(\frac{\ell_{1}^{2}+\ell_{2}^{2}}{2},\frac{n_{1}^{2}+n_{2}^{2}}{2},n_{1}\ell_{1}+n_{2}\ell_{2},\ell_{1}n_{2}-\ell_{2}n_{1},\ell_{3},n_{3}).
\]
 These invariants have the following closed Poisson structure 
\[
B_{\mu}=\left(\begin{array}{cccccc}
0 & z_{4}z_{6} & z_{4}z_{5} & 2z_{1}z_{6}-z_{3}z_{5} & 0 & -z_{4}\\
-z_{4}z_{6} & 0 & 0 & -2z_{2}z_{6} & 0 & 0\\
-z_{4}z_{5} & 0 & 0 & -2z_{2}z_{5} & 0 & 0\\
z_{3}z_{5}-2z_{1}z_{6} & 2z_{2}z_{6} & 2z_{2}z_{5} & 0 & 0 & -2z_{2}\\
0 & 0 & 0 & 0 & 0 & 0\\
z_{4} & 0 & 0 & 2z_{2} & 0 & 0
\end{array}\right)
\]
 with Casimirs $\mathcal{C}_{1}=\left|\bm{n}\right|^{2},\mathcal{C}_{2}=\ell_{3},\mathcal{C}_{3}=\frac{1}{2}\left(z_{3}^{2}+z_{4}^{2}\right)-2z_{1}z_{2}=0$
and $\mathcal{C}_{4}=\bm{n}\cdot\bm{\ell}=a$. For a fixed value of
$b=\ell_{3}$ we pick the following $3$ invariants $(x,y,z)=(n_{3},\ell_{1}n_{2}-\ell_{2}n_{1},\frac{\ell_{1}^{2}+\ell_{2}^{2}}{2}).$
These have the following Poisson structure

\[
B_{F}=\left(\begin{array}{ccc}
0 & 1-x^{2} & y\\
x^{2}-1 & 0 & b(a-bx)-2xz\\
-y & 2xz-b(a-bx) & 0
\end{array}\right)
\]
with Casimir 
\[
S=\frac{1}{2}(a-bx)^{2}-\left(1-x^{2}\right)z+\frac{y^{2}}{2}=0.
\]
In the reduced coordinates $\mathcal{H}$ can be expressed as 
\[
\mathcal{H}=\frac{b^{2}-a^{2}}{2I_{1}}+\frac{a^{2}}{2I_{3}}+\frac{z}{I_{1}}+kx^{2}+\chi x.
\]
} 

\bibliography{all.bib,hd.bib}
\bibliographystyle{plain}

\end{document}